\documentclass[sn-mathphys,Numbered]{sn-jnl}

\usepackage{graphicx}%
\usepackage{caption}
\usepackage{multirow}%
\usepackage{amsmath,amssymb,amsfonts}%
\usepackage{amsthm,amssymb}%
\usepackage{mathrsfs}%
\usepackage[title]{appendix}%
\usepackage{xcolor}%
\usepackage{textcomp}%
\usepackage{manyfoot}%
\usepackage{booktabs}%
\usepackage{algorithm}%
\usepackage{algorithmicx}%
\usepackage{enumitem}
\usepackage{algpseudocode}%
\usepackage{listings}%
\usepackage{mathrsfs}
\usepackage{arydshln}
\usepackage{threeparttable}
\usepackage{xcolor}
\usepackage{bm}
\newcommand{\ket}[1]{\left|#1\right\rangle}

\theoremstyle{thmstyleone}%
\newtheorem{theorem}{Theorem}
\theoremstyle{thmstyletwo}%
\newtheorem{remark}{Remark}%
\newtheorem{lemma}{Lemma}

\theoremstyle{thmstylethree}%

\begin{document}

\title[Article Title]{Construction of a Class of Communication-Efficient Quantum Secret Sharing Schemes }


\author[1]{\fnm{Chenhao} \sur{Li}}\email{mouse011126@163.com}

\author*[1]{\fnm{Zhihui} \sur{Li}}\email{lizhihui@snnu.edu.cn}

\author[1]{\fnm{Jiansheng} \sur{Guo}}\email{guojs@snnu.edu.cn}

\author[1]{\fnm{Yixin} \sur{Chen}}\email{chenyixin1997@snnu.edu.cn}

\author[1]{\fnm{Yewei} \sur{Wang}}\email{18766217130@163.com}

\affil*[1]{\orgdiv{School of Mathematics and Statistics}, \orgname{Shaanxi Normal University}, \orgaddress{\city{Xi'an}, \postcode{710119}, \country{China}}}


\abstract{Quantum secret sharing is a fundamental technique in quantum cryptography. 
However, in practical quantum networks, it still faces several bottlenecks, 
such as high quantum communication cost and low transmission efficiency. 
To reduce the communication cost, ramp quantum secret sharing schemes have 
been proposed in existing studies. Nevertheless, intermediate sets in such 
schemes may leak partial information about the secret. To address this problem, this paper presents a method for detecting 
intermediate sets. Based on this method, we further propose a 
communication-efficient perfect quantum secret sharing scheme with 
eavesdropping detection capability.We analyze the communication cost under different numbers 
of participating parties and determine the range of participants that minimizes 
the reconstruction communication cost. The results verify the communication 
efficiency and security of the proposed scheme.
}




\maketitle

\section{Introduction}\label{sec1}
Quantum secret sharing (QSS) is a class of quantum cryptographic protocols in which a quantum secret is encoded into a multipartite quantum state and distributed among several participants as quantum shares. The security of such protocols is rooted in fundamental quantum-mechanical constraints, including entanglement and the no-cloning theorem, which ensure that only authorized subsets can reconstruct the secret while unauthorized subsets obtain no useful information.
 Moreover, QSS is inherently 
sensitive to eavesdropping and tampering, thereby transforming the risk of 
single-point compromise in scenarios such as distributed key management, joint 
control, and secure computation into a security mechanism in which leakage is 
possible only through collusion among multiple parties.

According to the extent to which a participant set can recover the secret, quantum 
secret sharing schemes can be divided into two main categories: perfect quantum 
secret sharing schemes (PQSS)~\cite{ref1,ref18,ref20,ref21} and non-perfect quantum 
secret sharing schemes (Non-PQSS)~\cite{ref22,ref23,ref26,ref38}. A perfect quantum secret sharing scheme requires a strict separation between authorized and unauthorized participant sets. Specifically, the shares distributed by the dealer must enable any authorized set to recover the secret, while any unauthorized set remains information-theoretically independent of the secret. When this separation fails, namely when some non-authorized set can access a nonzero amount of information without being able to reconstruct the whole secret, the scheme becomes non-perfect, and such a set is commonly regarded as an intermediate set.

In quantum secret sharing, the shared secret may be either classical information 
encoded in quantum states~\cite{ref24,ref32,ref33} or a quantum state 
itself~\cite{ref2,ref36,ref37,ref39}. This paper focuses on QSS schemes for sharing 
quantum secrets. A QSS scheme for sharing quantum secrets generally consists of 
the following stages:
(1) Preparation and encoding of secret quantum state: The dealer prepares a secret quantum state and encodes it into 
$n$ quantum shares.
(2) Distribution of quantum shares: Each share is sent to the participants via quantum channels.
(3) Detection: Assess whether there is interference in the channels or from the participants by means of spot-checking, measurement, etc.
(4) Secret reconstruction and recovery: The authorized set gathers their shares to a combiner, who  performs decoding operations to recover the original secret quantum state.

In PQSS schemes for sharing quantum states, quantum threshold sharing schemes 
play a particularly important role. Such schemes are usually denoted by 
$\left(\left(t,n\right)\right)$-QTS. Let the set of participants be 
$P=\{P_1,P_2,\ldots,P_n\}$, where $P_1,P_2,\ldots,P_n$ may also be abbreviated as 
$\{1,2,\ldots,n\}$. In this setting, any subset of $P$ containing at least $t$ 
participants is a qualified set, whereas any subset containing fewer than $t$ 
participants is a forbidden set. In 1999, Cleve et al.~\cite{ref2} first pointed 
out that, in a $\left(\left(t,n\right)\right)$-QTS scheme, the size of each share, 
measured in qudits, cannot be smaller than the size of the secret itself. 
Consequently, when the combiner has to collect complete shares from $t$ 
participants in order to reconstruct the secret, the quantum communication cost 
in the reconstruction phase is at least $t$ times the size of the secret.

From 2019 to 2024, Senthoor et al. systematically investigated the communication 
cost of quantum secret sharing for quantum secrets~\cite{ref3,ref4,ref5}.
They defined the communication cost as the total amount of quantum information transmitted by authorized participants to the combiner during secret recovery, and used $CC_n(d)$ to denote the maximum communication cost over all authorized sets of size $d$. Based on this metric, Ref.~\cite{ref3} introduced communication-efficient quantum threshold secret sharing (CE-QTS), where the combiner contacts $d>t$ participants and downloads only partial shares from each participant. This significantly reduces the reconstruction cost compared with conventional threshold recovery. This idea was further generalized in several directions. Ref.~\cite{ref4} proposed a universal CE-QTS framework that allows a single threshold-sharing encoding to support different reconstruction sizes $d\ge t$. Ref.~\cite{ref5} established a unified information-theoretic framework for communication-efficient QTS, including lower bounds, achievability conditions, and comparisons for different authorized-set sizes. Other communication-efficient QSS constructions have also been proposed. Rui \emph{et al.}~\cite{ref9} constructed a communication-efficient $((t,n))$ QTS scheme using multipartite entangled states, where the corresponding unitary operations are characterized by invertible matrices over finite fields. Zhou \emph{et al.}~\cite{ref10} proposed an efficient $((t,n))$ PQSS scheme based on multipartite high-dimensional entangled states, in which the reconstruction communication cost decreases as the number $d$ of contacted authorized participants increases.

Since perfect quantum secret sharing schemes usually require relatively large 
quantum communication and storage costs, Non-PQSS schemes, by contrast, can 
achieve higher communication efficiency with smaller share sizes or lower 
reconstruction communication costs. Motivated by this idea, communication-efficient 
Non-PQSS schemes have attracted increasing attention in recent years~\cite{ref11,ref12,ref13,ref14}.
In a Non-PQSS scheme, there may exist participant sets that are neither qualified 
nor forbidden. Such sets are referred to as intermediate sets. Although the 
participants in an intermediate set cannot fully reconstruct the secret, they may 
still retrieve partial information about it. Therefore, the existence of 
intermediate sets may introduce potential security risks. More specifically, even 
though such sets are unable to recover the entire secret, they may extract partial information from the shares held by certain subsets of participants, thereby increasing the risk of secret leakage.

The security of intermediate sets in ramp-type quantum secret sharing, namely 
Non-PQSS, has been investigated from several perspectives. Zhang and 
Matsumoto~\cite{ref15} introduced the notion of strongly secure ramp quantum 
secret sharing, which requires that intermediate sets should not effectively 
obtain quantum information containing the essential content of the secret. 
Matsumoto~\cite{ref11} constructed ramp QSS schemes based on nested linear codes 
$C_2\subsetneq C_1\subseteq \mathbb{F}_q^n$ and provided criteria for determining 
qualified and forbidden sets. Subsequently, Matsumoto~\cite{ref14} studied ramp 
QSS from the viewpoint of quantum stabilizer codes, established necessary and 
sufficient conditions for qualified and forbidden sets, and quantified the amount 
of information that can be obtained by intermediate sets.

Ref.~\cite{ref6} constructed communication-efficient quantum secret sharing 
schemes based on extended CSS codes, covering both threshold and non-threshold 
settings. Inspired by the classical communication-efficient secret sharing 
constructions in Refs.~\cite{ref7,ref8}, this work employed generalized 
Reed--Solomon codes to construct CE-QTS schemes and further derived bounds on 
the communication cost of CE-QSS schemes in the secret reconstruction phase. 
In particular, the resulting CE-QTS schemes attain the optimal bounds in terms 
of both storage overhead and reconstruction communication cost.
In addition, Ref.~\cite{ref6} considered Non-PQSS schemes under the 
$\left((t,n;z)\right)$ access structure. Specifically, for 
$0\leq z<t\leq n$, any participant set $B\subseteq P$ with 
$\lvert B\rvert\geq t$ is a qualified set, whereas any participant set satisfying 
$\lvert B\rvert\leq z$ is a forbidden set. Under this access structure, apart 
from qualified and forbidden sets, there may exist participant sets whose sizes 
lie between these two regimes, namely sets satisfying 
$z<\lvert B\rvert<t$. Such sets are referred to as intermediate sets. They 
usually cannot fully reconstruct the secret, but may obtain partial information 
about it.
Based on this observation, Ref.~\cite{ref6} further incorporated the idea that 
contacting more participants can reduce the communication cost required for 
reconstruction. More precisely, a reconstruction parameter $d$ with 
$t\leq d\leq n$ was introduced into the $\left((t,n;z)\right)$ access structure, 
thereby yielding a $\left((t,n,d;z)\right)$ communication-efficient Non-PQSS 
scheme. It was shown that, when the combiner contacts $d$ participants, the 
communication cost satisfies $CC_n(d)<CC_n(t)$, and hence communication savings 
can be achieved in the secret reconstruction phase.

Motivated by the fact that, in the quantum secret sharing scheme of Ref.~\cite{ref6}, there exists a class of participant sets that belong to intermediate sets and may leak partial secret information, this paper investigates this issue in depth. The main contributions of this work are summarized as follows:
\begin{enumerate}[label={[\arabic*]}]
    \item 
   To address the potential information leakage caused by intermediate sets in 
existing communication-efficient QSS constructions based on extended CSS codes, 
this paper provides a sufficient condition, within the considered coding 
framework, for a class of participant sets to constitute intermediate sets. 
This condition can be used to analyze possible intermediate-set leakage under 
the $\left((t,n,d;z)\right)$ access structure, thereby providing a theoretical 
basis for the subsequent perfection procedure.

    \item 
    In the detection phase, we propose a stabilizer-based eavesdropping detection 
    mechanism compatible with extended CSS encoding. Specifically, during the 
    distribution phase, test code blocks are randomly selected, and either a 
    $Z$-test or an $X$-test is randomly performed. Each participant only needs to 
    carry out a single-qudit local measurement, and the measurement outcomes are 
    then classically aggregated to obtain the corresponding syndrome vector. A 
    decision is made according to the relationship between the empirical error rate 
    and a prescribed threshold. In this way, quantifiable and decidable detection of 
    $X$-type or $Z$-type disturbances can be achieved without relying on multi-qudit 
    joint measurements.

    \item 
    For a $\left((t,n,d;z)\right)$ non-perfect quantum secret sharing scheme containing 
intermediate sets, this paper introduces a classical $(t,n)$ threshold secret 
sharing mechanism to improve the original construction. By combining the 
potentially leaked information with an additional classical threshold-sharing 
layer, the ability of intermediate sets to obtain partial information about the 
quantum secret can be effectively eliminated. Consequently, this type of 
non-perfect quantum secret sharing scheme can be transformed into a perfect 
quantum secret sharing scheme.

    \item
    This paper also establishes a more comprehensive performance evaluation 
framework. In addition to analyzing the communication cost in the secret 
reconstruction phase, we further distinguish and quantify the communication 
costs in the distribution, detection, and reconstruction phases. This enables a 
comprehensive evaluation of the overall communication efficiency of the proposed 
scheme and provides a basis for unified comparison among different schemes.
\end{enumerate}

\section{Preliminaries}

This section fixes the notation and recalls the coding-theoretic concepts used
throughout the paper. These conventions will be used in the construction of the
extended CSS-based QSS scheme, the characterization of intermediate sets, and
the subsequent communication-cost analysis.

\subsection{Notation}

The main symbols used in this paper are collected in Table~\ref{tab1}. Unless
otherwise specified, all algebraic operations over vectors and matrices are
performed over the finite field $F_q$.

\begin{table}[h]
\centering
\caption{Notation used in this paper}\label{tab1}%
\begin{tabular}{cc}
\toprule
Symbol & Meaning\\
\midrule
$F_q$  
& Finite field with $q$ elements, where $q$ is an odd prime \\

$\{ |x\rangle : x\in F_q \}$    
& Computational basis of the single-qudit Hilbert space $\mathbb{C}^q$ \\

$|\underline{x}\rangle$, where 
$\underline{x}=(x_1,x_2,\ldots,x_n)\in F_q^n$   
& Tensor-product basis state $|x_1x_2\cdots x_n\rangle$ \\

$\{|\underline{x}\rangle : \underline{x}\in F_q^n\}$  
& Computational basis of the $n$-qudit space $\mathbb{C}^{q^n}$ \\

$[n]$  
& Index set $\{1,2,\ldots,n\}$ \\

$[i,j]$  
& Consecutive index set $\{i,i+1,\ldots,j\}$ \\

$V\in F_q^{m\times n}$  
& An $m\times n$ matrix whose entries lie in $F_q$ \\

$V_A$, where $A\subseteq [m]$  
& Submatrix obtained from $V$ by retaining the rows indexed by $A$ \\

$V^B$, where $B\subseteq [n]$ 
& Submatrix obtained from $V$ by retaining the columns indexed by $B$ \\

$V_A^B$, where $A\subseteq [m]$ and $B\subseteq [n]$ 
& Submatrix obtained by restricting $V$ to the rows in $A$ and the columns in $B$ \\
\botrule
\end{tabular}
\end{table}

\subsection{ communication cost}\label{subsect2}
\textbf{Definition 1.} According to the literature \cite{ref6},in a \(((t,n;z))\) QSS scheme, the communication cost associated with an authorized set \(A \subseteq [n]\) is defined as
\[
CC_n(A)=\sum_{j\in A} h_{j,A},
\]
where \(h_{j,A}\) denotes the number of particles sent by participant \(j\) to the combiner when the secret is reconstructed from the authorized set \(A\).

Here, we assume that, for a given authorized set, the portion of each accessed share that must be transmitted to the combiner is fixed in advance. This assumption is necessary because, for the same authorized set, there may exist different ways of partitioning the shares to enable secret recovery. Under this assumption, the communication cost of a QSS scheme can be characterized in a unified and unambiguous manner.

Accordingly, in a \(((t,n;z))\) QSS scheme, the communication cost for secret recovery from any \(d\) participants \((t \le d \le n)\) is defined as
\[
CC_n(d)=\max_{A\subseteq [n],\, |A|=d} CC_n(A).
\]

\subsection{CSS Codes and ECSS Codes}

For a linear code \(C\), the minimum distance is defined as
\[
wt(C)=\min\{wt(c)\mid c\in C,\; c\neq 0\}.
\]

Furthermore, for two linear codes \(C_1\) and \(C_0\) satisfying \(C_1\subseteq C_0\), the minimum distance of the nested code pair \((C_0,C_1)\) is defined as
\[
wt(C_0\setminus C_1)=\min\{wt(c)\mid c\in C_0,\; c\notin C_1\}.
\]

The CE-QSS construction used in this paper is built upon CSS codes and their extended structures. Therefore, we first briefly review the standard construction of Calderbank--Shor--Steane (CSS) codes, as well as their application to CE-QSS by Senthoor and Sarvepalli~\cite{ref6}.

\subsubsection{CSS Codes}

Let \(C_1\subseteq C_0\subseteq \mathbb{F}_q^n\) be two linear codes, and let \(\dim(C_i)=k_i\) for \(i\in\{0,1\}\). Since \(C_1\subseteq C_0\), the generator matrix of \(C_0\) may be written in the block form
\[
G_{C_0}=
\begin{bmatrix}
G_{C_0\setminus C_1}\\
G_{C_1}
\end{bmatrix},
\]
where \(G_{C_0\setminus C_1}\) generates a complementary subspace of \(C_1\) in \(C_0\).

For any logical quantum state \(|s\rangle\in \mathbb{F}_q^{\,k_0-k_1}\), the CSS encoding defined by the nested code pair \((C_0,C_1)\) can be written as
\[
|s\rangle \longmapsto \sum_{r\in \mathbb{F}_q^{k_1}}
\left|
\begin{bmatrix}
G_{C_0\setminus C_1}^{\,T} & G_{C_1}^{\,T}
\end{bmatrix}
\begin{bmatrix}
s\\
r
\end{bmatrix}
\right\rangle
=
\sum_{r\in \mathbb{F}_q^{k_1}}
\left|
G_{C_0}^{\,T}
\begin{bmatrix}
s\\
r
\end{bmatrix}
\right\rangle.
\]

Therefore, $\mathrm{CSS}(C_0, C_1)$ defines a quantum code with parameters
$\left[\!\left[n, k_0 - k_1\right]\!\right]_q$, whose minimum distance is
$\delta = \min \left\{ wt\!\left(C_0 \setminus C_1\right),\; wt\!\left(C_1^{\perp} \setminus C_0^{\perp}\right) \right\}.$From the perspective of quantum secret sharing, a quantum code with distance $\delta$
can correct up to $\delta - 1$ erasure errors. Accordingly, the QSS scheme induced by
this CSS code allows the secret to be recovered by any set of at least $n - \delta + 1$
participants, whereas any set of at most $\delta - 1$ participants can obtain no information about the secret.

\subsubsection{Construction of extended CSS codes (ECSS)}

  Let $G_E \in \mathbb{F}_q^{e \times n}$, and denote its row space by
$E \subseteq \mathbb{F}_q^n$. Following the construction of Senthoor and
Sarvepalli~[6], we require that $C_1 \subseteq C_0$ and $C_0 \cap E = \{0\}$.

The extended CSS code is defined as the CSS code induced by a pair of extended
codes $(D_0, D_1)$, whose generator matrices are given by
\[
G_{D_0} =
\begin{bmatrix}
G_{C_0} & 0 \\
G_E & I_e
\end{bmatrix},
\qquad
G_{D_1} =
\begin{bmatrix}
G_{C_1} & 0 \\
G_E & I_e
\end{bmatrix}.
\]

For the above definition of the extended CSS code, the matrix $G_E$ is not
required to be full rank. This code induces a QSS scheme with $n+e$ qudits.

The encoding of $\mathrm{ECSS}(C_0, C_1, G_E)$ can be written as
\[
|s\rangle \mapsto
\sum_{\substack{r_1 \in \mathbb{F}_q^{k_1} \\ r_2 \in \mathbb{F}_q^{e}}}
\left|
\begin{bmatrix}
G_{C_0 \setminus C_1}^{T} & G_{C_1}^{T} & G_E^{T} \\
0 & 0 & I_e
\end{bmatrix}
\begin{bmatrix}
s \\
r_1 \\
r_2
\end{bmatrix}
\right\rangle
=
\sum_{\substack{r_1 \in \mathbb{F}_q^{k_1} \\ r_2 \in \mathbb{F}_q^{e}}}
\left|
\begin{bmatrix}
G_{C_0 \setminus C_1}^{T} & G_{C_1}^{T} & G_E^{T}
\end{bmatrix}
\begin{bmatrix}
s \\
r_1 \\
r_2
\end{bmatrix}
\right\rangle
\otimes |r_2\rangle .
\]
\subsubsection{Stabilizers of CSS codes and their detection\cite{ref16}\cite{ref17}}
The parity-check matrix of the CSS code is defined as
\begin{equation}
\begin{bmatrix}
H(C_2^{\perp}) & 0 \\
0 & H(C_1)
\end{bmatrix}.
\label{eq.1}
\end{equation}

Let
$H_Z = H(C_2^{\perp}) = \bigl(h_{rj}^{(Z)}\bigr) \in \mathbb{Z}_q^{\,r_Z \times n},
\qquad
H_X = H(C_1) = \bigl(h_{rj}^{(X)}\bigr) \in \mathbb{Z}_q^{\,r_X \times n},$
which denote the $Z$-type and $X$-type parity-check matrices, respectively.
Then the corresponding stabilizer generators are given by

\begin{equation}
g_{Z,r}=\bigotimes_{j=1}^{n} Z_j^{\,h_{rj}^{(Z)}},\quad r=1,\ldots,r_Z,
\qquad
g_{X,r}=\bigotimes_{j=1}^{n} X_j^{\,h_{rj}^{(X)}},\quad r=1,\ldots,r_X.
\label{eq.2}
\end{equation}
For example, for the seven-qubit Steane code, the parity-check matrix given by Eq. (\ref{eq.1}) is
\[
H=
\left(
\begin{array}{ccccccc|ccccccc}
0 & 0 & 0 & 1 & 1 & 1 & 1 & 0 & 0 & 0 & 0 & 0 & 0 & 0 \\
0 & 1 & 1 & 0 & 0 & 1 & 1 & 0 & 0 & 0 & 0 & 0 & 0 & 0 \\
1 & 0 & 1 & 0 & 1 & 0 & 1 & 0 & 0 & 0 & 0 & 0 & 0 & 0 \\
0 & 0 & 0 & 0 & 0 & 0 & 0 & 0 & 0 & 0 & 1 & 1 & 1 & 1 \\
0 & 0 & 0 & 0 & 0 & 0 & 0 & 0 & 1 & 1 & 0 & 0 & 1 & 1 \\
0 & 0 & 0 & 0 & 0 & 0 & 0 & 1 & 0 & 1 & 0 & 1 & 0 & 1
\end{array}
\right).
\]
The corresponding stabilizer generators are shown in Table~\ref{tab:steane_stabilizers}.

\begin{table}[htbp]
\centering
\caption{Stabilizer generators corresponding to the seven-qubit Steane code}
\label{tab:steane_stabilizers}
\begin{tabular}{c|ccccccc}
\hline
\textbf{Operator } & \textbf{1} & \textbf{2} & \textbf{3} & \textbf{4} & \textbf{5} & \textbf{6} & \textbf{7} \\
\hline
$g_1$ & $I$ & $I$ & $I$ & $X$ & $X$ & $X$ & $X$ \\
$g_2$ & $I$ & $X$ & $X$ & $I$ & $I$ & $X$ & $X$ \\
$g_3$ & $X$ & $I$ & $X$ & $I$ & $X$ & $I$ & $X$ \\
$g_4$ & $I$ & $I$ & $I$ & $Z$ & $Z$ & $Z$ & $Z$ \\
$g_5$ & $I$ & $Z$ & $Z$ & $I$ & $I$ & $Z$ & $Z$ \\
$g_6$ & $Z$ & $I$ & $Z$ & $I$ & $Z$ & $I$ & $Z$ \\
\hline
\end{tabular}
\vspace{0.3em}
\small
Remark: Each row represents one group element, and each column indicates the action of the Pauli operators $I$, $X$, and $Z$ on the corresponding one of the seven qubits.
\end{table}
Let $|\psi\rangle$ be a quantum state in the code space, and let $g_{Z,r}$ and
$g_{X,s}$ denote the $Z$-type and $X$-type stabilizers of the code space,
respectively. Then, for all $r$ and $s$, one has
\[
g_{Z,r}|\psi\rangle = |\psi\rangle, \qquad
g_{X,s}|\psi\rangle = |\psi\rangle.
\]

Since the ECSS code is also a special form of a CSS code, one can construct the
parity-check matrices of the ECSS code and thereby derive its stabilizer generators.

If the local measurement outcomes in the computational basis form a vector
$\bm{x} = (x_1,\ldots,x_n) \in \mathbb{Z}_q^n,$
then the syndrome vector with respect to the $Z$-type parity-check matrix $H_Z$
is defined by $\bm{s}_Z = H_Z \bm{x}^{T},$
whose $r$th component is
\[
s_{Z,r} = \sum_{j=1}^{n} h_{rj}^{(Z)} x_j \pmod q.
\]

If the local measurement outcomes in the Fourier basis form a vector
$\bm{k} = (k_1,\ldots,k_n) \in \mathbb{Z}_q^n,$
then the syndrome vector with respect to the $X$-type parity-check matrix $H_X$
is defined by
$\bm{s}_X = - H_X \bm{k}^{T} \pmod q,$
whose $r$th component is
\[
s_{X,r} = - \sum_{j=1}^{n} h_{rj}^{(X)} k_j \pmod q.
\]
Therefore, for ideal error-free codeword test states, the syndromes satisfy
$\bm{s}_Z=\bm{0}$ and $\bm{s}_X=\bm{0}$. Hence, checking whether the
aggregated syndrome vector vanishes provides a criterion for verifying the
stabilizer constraints, and thus serves as a key indicator for detecting the
presence of disturbances.

\subsubsection{Construction Stage Based on the ECSS Quantum Circuit}
The quantum circuit for the ECSS code is introduced below; see
Fig.~\ref{fig:ecss-encoding-circuit}.

\begin{figure}[htbp]
    \centering
    \includegraphics[width=0.7\textwidth]{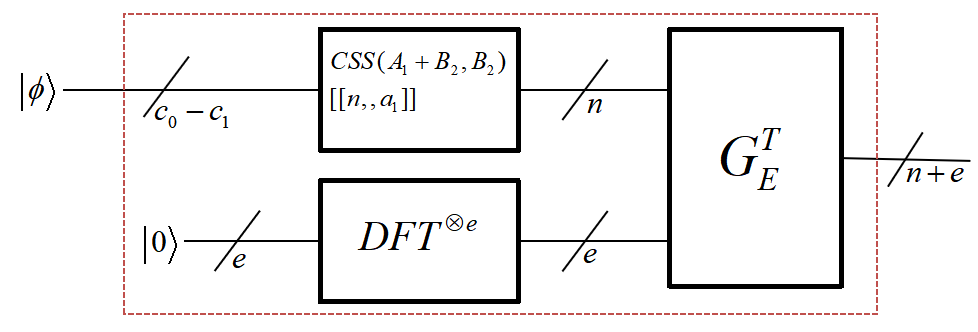}
    \caption{Quantum circuit structure of the ECSS code.}
    \label{fig:ecss-encoding-circuit}
\end{figure}

According to the definition in Section~2.3.2, the ECSS encoding can be divided
into three steps.

In the first step, the secret quantum state $|s\rangle$ is encoded according to $\mathrm{CSS}(C_0,C_1)$,
\[
|s\rangle \mapsto
\sum_{r_1 \in \mathbb{F}_q^{k_1}}
\left|
\begin{bmatrix}
G_{C_0 \setminus C_1}^{T} & G_{C_1}^{T}
\end{bmatrix}
\begin{bmatrix}
s \\
r_1
\end{bmatrix}
\right\rangle
=
\sum_{r_1 \in \mathbb{F}_q^{k_1}}
\left|
G_{C_0}^{T}
\begin{bmatrix}
s \\
r_1
\end{bmatrix}
\right\rangle .
\]

In the second step, ancilla qudits are introduced, and the operator $DFT^{\otimes e}$ is applied to their initial state $|0^{e}\rangle$, yielding the uniform superposition
\[
DFT^{\otimes e}\,|0^{e}\rangle
={\otimes_1 ^e}
\frac{1}{\sqrt{q^{e}}}
\sum_{r_2 \in \mathbb{F}_q^{e}} |r_2\rangle .
\]

In the third step, an invertible linear transformation implemented by $G_E^{T}$ is applied between the first encoded registers and the ancilla registers. Its action is given by
\[
|x\rangle |r_2\rangle \mapsto |x + G_E^{T}r_2\rangle \otimes |r_2\rangle,
\qquad
x \in \mathbb{F}_q^{n},\; r_2 \in \mathbb{F}_q^{e}.
\]

Ignoring the normalization factor, the encoding of the logical state $|s\rangle$ under $\mathrm{ECSS}(C_0,C_1,G_E)$ can therefore be written as
\[
\mathrm{ECSS}(C_0,C_1,G_E)\,|s\rangle
=
\sum_{\substack{r_1 \in \mathbb{F}_q^{\,k_1} \\ r_2 \in \mathbb{F}_q^{\,e}}}
\left|
\begin{bmatrix}
G_{C_0 \setminus C_1}^{T} & G_{C_1}^{T} & G_E^{T} \\
0 & 0 & I_e
\end{bmatrix}
\begin{bmatrix}
s \\
r_1 \\
r_2
\end{bmatrix}
\right\rangle .
\tag{3}
\label{eq.3}
\]

\subsection{Disentangling Effect of the ADD Gate}

In a $q$-dimensional quantum system, the secret recovery process requires the use of the generalized controlled-addition gate $\mathrm{ADD}$, defined by
\[
\mathrm{ADD}\,|i\rangle_c |j\rangle_t = |i\rangle_c |j+i\rangle_t,
\qquad i,j \in \mathbb{F}_q,
\]
where the addition is performed over $\mathbb{F}_q$. Its inverse gate is denoted by $\mathrm{ADD}^{\dagger}$, satisfying
\[
\mathrm{ADD}^{\dagger}\,|i\rangle_c |j\rangle_t = |i\rangle_c |j-i\rangle_t.
\]

We next illustrate the roles of $\mathrm{ADD}$ and $\mathrm{ADD}^{\dagger}$ in secret recovery through a simple example.

\medskip
\noindent\textbf{Example 1.} Consider the initial encoded state
\begin{equation}
|s_1,s_2\rangle
\mapsto
\sum_{r_1,r_2 \in \mathbb{F}_3}
|s_1+r_1,\,r_2\rangle_1
\otimes
|s_2+r_1,\,r_1+r_2\rangle_2
\otimes
|s_1+s_2+2r_1,\,r_1+r_2\rangle_3 .
\tag{4}
\label{eq.4}
\end{equation}

In the tensor-product state in Eq.~(\ref{eq.4}), the subscripts $1$, $2$, and $3$ attached 
to the factors indicate that the corresponding quantum registers are held by 
participants $P_1$, $P_2$, and $P_3$, respectively. Clearly, each participant holds 
two layers of share information. For example, the first-layer quantum share held 
by $P_1$ is $\left|s_1+r_1\right\rangle$, while the second-layer quantum share held 
by $P_1$ is $\left|r_2\right\rangle$.

First, we apply the inverse controlled-addition operation $\mathrm{ADD}^{\dagger}$ 
to the registers $\left|r_2\right\rangle_1$ and 
$\left|r_1+r_2\right\rangle_2$ in Eq.~(\ref{eq.4}), so as to eliminate the offset introduced 
by $r_2$ in the second-layer particle held by $P_2$, namely 
$\left|r_1+r_2\right\rangle_2$. This gives
\begin{equation}
\sum_{r_1,r_2 \in \mathbb{F}_3}
|s_1+r_1,\,r_2\rangle_1
\otimes
|s_2+r_1,\,r_1\rangle_2
\otimes
|s_1+s_2+2r_1,\,r_1+r_2\rangle_3 .
\tag{5}
\label{eq.5}
\end{equation}

Next, we further apply the inverse controlled-addition operation 
$\mathrm{ADD}^{\dagger}$ to eliminate the offsets introduced by $r_1$ in the 
registers held by $P_1$ and $P_2$ in Eq.~(\ref{eq.5}), namely 
$\left|s_1+r_1\right\rangle_1$ and $\left|s_2+r_1\right\rangle_2$. 
Then Eq.~(\ref{eq.5}) can be transformed into
\begin{equation}
\sum_{r_1,r_2 \in \mathbb{F}_3}
|s_1,\,r_2\rangle_1
\otimes
|s_2,\,r_1\rangle_2
\otimes
|s_1+s_2+2r_1,\,r_1+r_2\rangle_3 .
\tag{6}
\label{eq.6}
\end{equation}

Although the local registers of $P_1$ and $P_2$ now appear in the form
$|s_1\rangle$ and $|s_2\rangle$, the overall system is still not fully
disentangled, because the auxiliary registers still contain the secret-dependent
term $s_1+s_2+2r_1$. If the auxiliary subsystem were measured at this stage, the
global state would collapse, thereby destroying the coherence of an arbitrary
superposition secret. Therefore, proper secret reconstruction requires the final
state to have a tensor-product form
\[
|\Psi\rangle = \bigl(|s_1\rangle_1 |s_2\rangle_2\bigr) \otimes |\phi\rangle,
\]
where the auxiliary state $|\phi\rangle$ no longer depends on the secret.

To achieve this, one continues to apply a sequence of $\mathrm{ADD}$ and
$\mathrm{ADD}^{\dagger}$ gates to eliminate the remaining random terms through
reversible linear elimination. The final state becomes
\begin{equation}
|s_1\rangle_1 |s_2\rangle_2
\otimes
\sum_{t_1,t_2\in\mathbb{F}_3}
|t_1\rangle_1 |t_2\rangle_2 |t_1\rangle_3 |t_2\rangle_3 .
\tag{7}
\label{eq.7}
\end{equation}
where $t_1 = s_1 + s_2 + 2r_1, \qquad t_2 = r_1 + r_2.$

At this point, the secret component $|s_1\rangle_1 |s_2\rangle_2$ in Eq.(\ref{eq.7}) has been completely disentangled from the auxiliary registers. Hence, by successively canceling the translation terms introduced by the random variables, the recovery process disentangles the secret registers from the auxiliary registers and ensures coherent recovery for an arbitrary superposition secret, thereby achieving correct quantum secret reconstruction.

\section{Detection of intermediate sets	}\label{sec3}
In this section, we investigate the communication-efficient quantum secret sharing construction based on generalized Reed-Solomon codes given in Corollary 4 of ~\cite{ref6} . On this basis, we analyze the leakage phenomenon caused by intermediate sets and provide a sufficient criterion for their identification.

\subsection{An example demonstrating the existence of an intermediate set}
We first present an explicit example of an intermediate set. According to Corollary~4 in~ ~\cite{ref6}, let $(t,n,d,z)=(4,6,5,2).$ We will show that, in this scheme, the participant set
$J=\{P_1,P_2,P_4\}$
forms an intermediate set and hence leaks partial secret information.

\medskip
\noindent\textbf{Example 2.} The first-layer particles obtained by the combiner are
\begin{align*}
    \sum_{r\in \mathbb{F}_7^{4}}
    |s_1+s_3+s_5+r_1+r_3\rangle
    |s_1+2s_3+4s_5+r_1+2r_3\rangle
    |s_1+4s_3+2s_5+r_1+4r_3\rangle\\
    |s_2+s_4+s_6+r_2+r_4\rangle
    |s_2+2s_4+4s_6+r_2+2r_4\rangle
    |s_2+4s_4+2s_6+r_2+4r_4\rangle 
\tag{8}
\label{eq.8}
\end{align*}

and the second-layer particles are
\begin{equation}
\sum_{r\in \mathbb{F}_7^{4}}
|r_3+r_4+r_5+r_6\rangle
|2r_3+4r_4+r_5+2r_6\rangle
|4r_3+2r_4+r_5+4r_6\rangle .
\tag{9}
\label{eq.9}
\end{equation}

After applying multiple ADD-type operations, the first-layer quantum state can be transformed into
\[
\sum_{r\in \mathbb{F}_7^{4}}
|s_2+s_4+s_6+r_2\rangle
|s_4\rangle
|6s_6\rangle .
\]

At this stage, the quantum state already reveals partial secret information, namely $|s_4\rangle |6s_6\rangle .$ Such information may be further exploited to infer secret-related content contained in the first-layer shares. Therefore, it is evident that
$J=\{P_1,P_2,P_4\}$
constitutes an intermediate set. We next present a criterion for identifying such intermediate sets.

\subsection{Criterion for Identifying Intermediate Sets}

In the ECSS encoding considered in this section, the share held by each participant
can be expressed as a linear combination of the secret variables and the random
variables.

\begin{theorem}
    For any participant set $B$, after applying the ECSS encoding in Eq.(\ref{eq.3}), its joint
share can be written as
$\bm{y}_B = M_B
\begin{bmatrix}
\bm{s} \\
\bm{r}
\end{bmatrix},$
where $\bm{s}$ denotes the secret vector and $\bm{r}$ denotes the random vector.
Write
\[
M_B=\left[M_B^{(s)} \mid M_B^{(r)}\right],
\]
where $M_B^{(s)}$ and $M_B^{(r)}$ are the coefficient matrices corresponding to the
secret part and the random part, respectively. If there exists an invertible matrix
$T$ such that some row of $T M_B$ has the form
$(0,\ldots,0,$a$,0,\ldots,0 \mid 0,\ldots,0),
where \ $a$\neq 0,$
then the set $B$ is an intermediate set.
\label{thm1}
\end{theorem}

\begin{proof}
Since the generalized $\mathrm{ADD}$ gate and its inverse correspond, at the
coefficient level, to invertible linear transformations, the application of a
sequence of local reversible operations to the joint share transforms the
corresponding coefficient matrix into $T M_B$.
If one row of $T M_B$ is of the form
$(0,\ldots,0,a,0,\ldots,0 \mid 0,\ldots,0),$
where $a\neq 0$, then the quantum register corresponding to this row depends
only on the secret component $a s_i$ and is independent of all random variables.
Since all operations are performed over the finite field $\mathbb{F}_q$ and
$a\neq 0$, the coefficient $a$ is invertible in $\mathbb{F}_q$. Therefore, this
register contains nontrivial information about the secret component $s_i$.
Consequently, the participant set $B$ is an intermediate set.
\end{proof}

Since the generalized $\mathrm{ADD}$ gate and its inverse correspond, at the
coefficient level, to reversible elementary row operations, one can determine
whether a participant set $B$ leaks secret information by performing reversible
linear elimination on the coefficient matrix $M_B$. To illustrate this procedure,
we rewrite Example~2 in matrix form.

Accordingly, the first-layer particles in Eq.(\ref{eq.8}) can be rewritten as
\begin{equation}
\sum_{r\in\mathbb{F}_7^4}
\begin{bmatrix}
1 & 1 & 1 & 1 & 1 \\
1 & 2 & 4 & 1 & 2 \\
1 & 4 & 2 & 1 & 4
\end{bmatrix}
\begin{bmatrix}
s_1 & s_2 \\
s_3 & s_4 \\
s_5 & s_6 \\
r_1 & r_2 \\
r_3 & r_4
\end{bmatrix}
\tag{10}
\label{eq.10}
\end{equation}

Similarly, the second-layer particles in Eq.(\ref{eq.9}) can be rewritten as
\[
\sum_{r\in\mathbb{F}_7^4}
\begin{bmatrix}
1 & 1 & 1 & 1 \\
2 & 4 & 1 & 2 \\
4 & 2 & 1 & 4
\end{bmatrix}
\begin{bmatrix}
r_3 \\
r_4 \\
r_5 \\
r_6
\end{bmatrix}
\]

After applying several $\mathrm{ADD}$ and $\mathrm{ADD}^{\dagger}$ gates to the
second-layer particles, reversible linear elimination can be performed on the
random-variable part. The second layer can then be transformed into the equivalent
form
\[
\sum_{r\in\mathbb{F}_7^4}
\begin{bmatrix}
1 & 1 & 1 & 1 \\
0 & 2 & 0 & 0 \\
0 & 0 & 3 & 0
\end{bmatrix}
\begin{bmatrix}
r_3 \\
r_4 \\
r_5 \\
r_6
\end{bmatrix}
\]

Through such joint operations, one can separate from the second-layer shares
independent linear combinations involving the random variables $r_4$ and $r_5$.
Furthermore, by using the obtained state $|2r_4\rangle$ from the second layer and
applying the corresponding $\mathrm{ADD}$ operations to eliminate the associated
random terms in Eq.(~\ref{eq.10}), one obtains
\[
\sum_{r\in\mathbb{F}_7^4}
\begin{bmatrix}
1 & 1 & 1 & 1 & 0 \\
0 & 1 & 0 & 0 & 0 \\
0 & 0 & 6 & 0 & 0
\end{bmatrix}
\begin{bmatrix}
s_2 \\
s_4 \\
s_6 \\
r_2 \\
r_4
\end{bmatrix}
\]
Although the set $J=\{1,2,4\}$ is insufficient to reconstruct the entire secret,
its joint state allows one to explicitly extract linear information about the
secret components $s_4$ and $s_6$. In other words, this set is not a completely
uninformed unauthorized set, but rather a typical intermediate set. This motivates
the following detection criterion.

\subsection{Intermediate-Set Detection Theorem}

To analyze the leakage caused by intermediate sets in the scheme corresponding
to Corollary~4 of Ref.\cite{ref6}, we represent the quantum shares held by the participants
as linear combinations of the secret variables and the random variables, and then
perform elimination on the relevant coefficient matrices by exploiting the
invertible linear transformations induced by the $\mathrm{ADD}$ gate and its inverse.

For convenience, we first define the complete homogeneous symmetric polynomial.
For $s\geq 0$,
\[
h_s(y_0,\ldots,y_m)
:=
\sum_{i_0+\cdots+i_m=s}
y_0^{i_0}\cdots y_m^{i_m}.
\]
In particular,
$h_0(y_0,\ldots,y_m)=1.$
It satisfies the recurrence relation
\begin{equation}
h_s(y_0,\ldots,y_m)
=
h_s(y_0,\ldots,y_{m-1})
+
y_m h_{s-1}(y_0,\ldots,y_m),
\qquad s\geq 1.
\tag{11}
\label{eq.11}
\end{equation}

For $m\geq 0$, define
\[
N_0(x)=1,
\qquad
N_m(x)=\prod_{r=0}^{m-1}(x-x_r),
\qquad m\geq 1.
\]

and
\begin{equation}
\phi_m(x) := \frac{x^k N_m(x)}{x_m^k N_m(x_m)}, \qquad 0 \le m \le n.
\tag{12}
\label{eq.12}
\end{equation}
\begin{lemma}\label{lemma1}
     For any $n \ge 0$, the following identity holds:
\begin{equation}
x^n = \sum_{m=0}^{n} h_{n-m}(x_0,\ldots,x_m)\,N_m(x).
\tag{13}
\label{eq.13}
\end{equation}
\end{lemma}
\begin{proof}
    When $n=0$, since $h_0(x_0)=1$ and $N_0(x)=1$, Eq.~(\ref{eq.13}) clearly holds.

Assume that Eq.(\ref{eq.13}) holds for $n=j$. Consider the case $n=j+1$. By the induction hypothesis,
\[
x^j = \sum_{m=0}^{j} h_{j-m}(x_0,\ldots,x_m)\,N_m(x).
\]
Multiplying both sides by $x$, we obtain
\[
x^{j+1} = \sum_{m=0}^{j} h_{j-m}(x_0,\ldots,x_m)\,x\,N_m(x).
\]
Using the identity $x=(x-x_m)+x_m$, we have
\begin{equation}
x^{j+1}
=
\sum_{m=0}^{j} h_{j-m}(x_0,\ldots,x_m)\,N_{m+1}(x)
+
\sum_{m=0}^{j} x_m\,h_{j-m}(x_0,\ldots,x_m)\,N_m(x).
\tag{14}
\label{eq.14}
\end{equation}

Letting $r=m+1$, Eq.(\ref{eq.14}) can be rewritten as
\begin{equation}
x^{j+1}
=
x_0 h_j(x_0) N_0(x)
+
\sum_{r=1}^{j} \bigl(h_{j+1-r}(x_0,\ldots,x_{r-1})
+ x_r h_{j-r}(x_0,\ldots,x_r)\bigr) N_r(x)
+ N_{j+1}(x).
\tag{15}
\label{eq.15}
\end{equation}

By Eq.~(\ref{eq.11}) and the fact that $h_0(x_0,\ldots,x_{j+1})=1$, we further obtain
\begin{equation}
x^{j+1}
=
h_{j+1}(x_0) N_0(x)
+
\sum_{r=1}^{j} h_{j+1-r}(x_0,\ldots,x_r) N_r(x)
+ N_{j+1}(x).
\tag{16}
\label{eq.16}
\end{equation}

Hence,
\[
x^{j+1}
=
\sum_{m=0}^{j+1} h_{j+1-m}(x_0,\ldots,x_m)\,N_m(x).
\]
Therefore, Eq.(\ref{eq.13}) holds for all $j \ge 0$ by mathematical induction.
\end{proof}

\begin{lemma}\label{lem:lu-decomposition}
    Let
\[
V = V(x_0,\ldots,x_n; k) := \bigl(x_i^{k+j}\bigr)_{0 \le i,j \le n}.
\]
If $x_0,\ldots,x_n \in \mathbb{F}_q$ are pairwise distinct, and if $k>0$, then $V$
admits an LU-decomposition
\[
V = LU,
\]
where $L$ is a unit lower triangular matrix and $U$ is an upper triangular matrix.
\end{lemma}
\begin{proof}
    By Lemma~1, multiplying both sides of Eq.(\ref{eq.13}) by $x^k$ yields
\[
x^{k+j} = \sum_{m=0}^{j} h_{j-m}(x_0,\ldots,x_m)\,x^k N_m(x).
\]
By Eq.~(\ref{eq.12}), $x^k N_m(x) = x_m^k N_m(x_m)\,\phi_m(x)$, and hence
\[
x^{k+j}
=
\sum_{m=0}^{j} h_{j-m}(x_0,\ldots,x_m)\,x_m^k N_m(x_m)\,\phi_m(x).
\]

Define the entries of $L$ by
\begin{equation}
l_{im} :=
\begin{cases}
0, & i < m,\\[4pt]
\phi_m(x_i)
= \dfrac{x_i^k \prod_{r=0}^{m-1}(x_i - x_r)}
{x_m^k \prod_{r=0}^{m-1}(x_m - x_r)}, & i \ge m,
\end{cases}
\tag{17}
\label{eq.17}
\end{equation}
and define the matrix $U = (u_{mj})_{0 \le m,j \le n}$ by
\begin{equation}
u_{mj} :=
\begin{cases}
0, & m > j,\\
x_m^k N_m(x_m)\, h_{j-m}(x_0,\ldots,x_m), & m \le j.
\end{cases}
\tag{18}
\label{eq.18}
\end{equation}

Setting $x = x_i$, we obtain
\begin{equation}
x_i^{k+j}
=
\sum_{m=0}^{j} u_{mj}\,\phi_m(x_i)
=
\sum_{m=0}^{n} l_{im} u_{mj}
= (LU)_{ij}.
\tag{19}
\label{eq.19}
\end{equation}

Thus, $V=LU$.
where
$U = (u_{mj})_{0 \le m,j \le n}, 
L = (l_{im})_{0 \le i,m \le n}.$
\end{proof}

\begin{lemma}
[\textnormal{[6, Corollary~4]}]\label{lem:ce-qss-construction}
Let $q \ge n+1$ be a prime. For any $0<z<t<d\le n=t+z$, let
$B_0,B_1,B_2,A_1,A_2,E$ be generalized Reed--Solomon codes used for encoding,
where $B_0$ is generated by $V(x_0,\ldots,x_n;k)=(x_i^{k+j})_{0\le i,j\le n}$.
Set
\[
b_0=d,\quad b_1=z,\quad b_2=z-d+t,\quad e=d-t,\quad a_1=d-z,\quad a_2=t-z.
\]
Then one obtains a quantum secret sharing (QSS) scheme with parameters
$\bigl((t,n=t+z; d; z)\bigr)_q$.
\end{lemma}

\begin{theorem}\label{thm2}
For the quantum secret sharing scheme given in Lemma~\ref{lem:ce-qss-construction},
intermediate sets may leak partial information about the quantum secret.
When the secret is encoded in computational basis states, if
\begin{equation}
q \,\Bigg|\,
\sum_{i_0+i_1+\cdots+i_t=l}
x_0^{i_0}x_1^{i_1}\cdots x_t^{i_t},
\qquad
l=1,2,\ldots,n-t,
\tag{20}
\label{eq.20}
\end{equation}
then the scheme leaks information about at least one secret component.
\end{theorem}

\begin{proof}
According to Eq.~(\ref{eq.19}) in Lemma~\ref{lem:lu-decomposition}, the
Vandermonde-type matrix
\begin{equation*}
V(x_0,x_1,\ldots,x_n)
=
\begin{bmatrix}
x_0^k & x_0^{k+1} & \cdots & x_0^{k+t} \\
x_1^k & x_1^{k+1} & \cdots & x_1^{k+t} \\
\vdots & \vdots & \ddots & \vdots \\
x_n^k & x_n^{k+1} & \cdots & x_n^{k+t}
\end{bmatrix}
\end{equation*}
admits an LU decomposition
\begin{equation}
V=LU,
\qquad
V_{ij}=\sum_{m=0}^{n}L_{im}U_{mj},
\tag{21}
\label{eq.21}
\end{equation}
where \(L=(L_{im})_{0\le i,m\le n}\) and
\(U=(U_{mj})_{0\le m,j\le n}\) are given by
Eqs.(\ref{eq.17}) and~(\ref{eq.18}), respectively.

Since \(L\) is a unit lower triangular matrix, it can be expressed as a
finite product of elementary matrices of the third type. Moreover, performing
a third-type elementary row operation on the coefficient matrix corresponding
to the particles is equivalent to applying an \(\operatorname{ADD}\) gate to
the corresponding particles. Therefore, after carrying out the corresponding
\(\operatorname{ADD}\) operations, the coefficient matrix associated with the
particles can be transformed into the following upper triangular form:
\[
\resizebox{\textwidth}{!}{$
\begin{bmatrix}
x_0^k
&
x_0^{k+1}
&
x_0^{k+2}
&
\cdots
&
x_0^{k+n}
\\[1mm]
0
&
x_1^k(x_1-x_0)
&
x_1^k(x_1-x_0)(x_1+x_0)
&
\cdots
&
x_1^k(x_1-x_0)
\displaystyle\sum_{i+j=n-1}x_0^i x_1^j
\\[2mm]
0
&
0
&
x_2^k(x_2-x_0)(x_2-x_1)
&
\cdots
&
x_2^k(x_2-x_0)(x_2-x_1)
\displaystyle\sum_{i+j+m=n-2}x_0^i x_1^j x_2^m
\\
\vdots
&
\vdots
&
\vdots
&
\ddots
&
\vdots
\\[1mm]
0
&
0
&
0
&
\cdots
&
x_n^k(x_n-x_0)(x_n-x_1)\cdots(x_n-x_{n-1})
\end{bmatrix}
$}
\]

Under the condition
\[
q \,\Bigg|\,
\sum_{i_0+i_1+\cdots+i_t=l}
x_0^{i_0}x_1^{i_1}\cdots x_t^{i_t},
\qquad
l=1,2,\ldots,n-t,
\]
all the entries to the right of the diagonal entry in the \(t\)-th row vanish,
whereas the diagonal entry remains nonzero because
\(x_0,x_1,\ldots,x_t\) are pairwise distinct and nonzero. Consequently, the
particle corresponding to this diagonal entry contains nontrivial information
about the quantum secret and is independent of the random variables.

From the above discussion, starting from the $(t+1)$-th row, each row of the 
matrix contains only one nonzero entry, namely its diagonal entry, while all the 
other entries are zero. Therefore, these rows satisfy the condition in 
Theorem~\ref{thm1} and hence reveal information about the quantum particles corresponding 
to the respective nonzero diagonal entries.
\end{proof}

\section{Construction of an Efficient  PQSS}

In this section, we improve a non-PQSS scheme with intermediate sets into
PQSS by incorporating a classical secret sharing mechanism.

\subsection{System Model}

For the CE-QSS given in Lemma~3,
we assume that
\[
q \mid
\sum_{i_0+i_1+\cdots+i_t=l}
x_0^{i_0}x_1^{i_1}\cdots x_t^{i_t},
\qquad
l=1,2,\ldots,n-t.
\]

By Theorem \ref{thm2}, the proposed scheme necessarily contains intermediate sets. 
To eliminate this security risk, this section introduces a classical threshold 
secret sharing mechanism in addition to the original quantum sharing scheme. 
The classical layer is used to protect the random key employed for pre-encrypting 
the quantum secret. In this way, the Non-PQSS scheme can be transformed into a 
PQSS scheme, thereby preventing partial participants from recovering partial 
information about the secret and hence avoiding secret leakage.

In this section, we construct a class of hybrid quantum secret sharing schemes 
over the finite field $\mathbb{F}_q$, where $q>n$ is a prime and all operations 
are performed over $\mathbb{F}_q$. The proposed scheme adopts a pre-specified 
reconstruction-set model similar to that in Ref.~\cite{ref9}.In this work, we assume that the authorized subset used for reconstruction is determined prior to the distribution stage. That is, before Alice distributes the encoded quantum systems, the combiner Peter specifies the set of participants that will take part in the recovery process in the current round. Under this pre-announced reconstruction model, Alice distributes only the quantum shares associated with the specified authorized subset, whereas the remaining quantum systems are retained by Alice.

Theoretically, $n$ groups of encoded shares are still generated. However, only 
the shares belonging to the pre-specified authorized set $J$ are actually 
transmitted to the participants in $J$, whereas the remaining shares do not 
enter the distribution phase of the current round. The main purpose of this 
setting is to reduce the quantum communication cost in the distribution phase. 
Therefore, the subsequent analysis of the communication cost in the distribution 
phase and the overall communication efficiency is based on this model 
assumption.

The proposed scheme involves a dealer Alice, $n$ participants, and a combiner 
Peter. Alice is responsible for encoding and distributing the secret quantum 
state, and for inserting test code blocks during the protocol execution to 
enable security detection. Each participant receives and stores the corresponding 
quantum shares and classical shares. In the secret reconstruction phase, Peter 
collects the required shares from the pre-specified authorized set and recovers 
the original quantum secret by performing the corresponding unitary operations. 
For ease of implementation, we further assume that authenticated classical 
channels are available among the participants and between the participants and 
the combiner, so that the necessary auxiliary information can be exchanged. We 
also assume that the participants are capable of performing the required 
measurements according to the protocol.

\subsection{Preparation Phase}

\subsubsection{Parameter Setting}

Let $t$ denote the minimum number of participants required to reconstruct the secret, and let $d$ denote the number of participants used for efficient secret recovery, where $t \le d \le n = t + z.$
Define
\[
b_0 = d,\quad b_1 = z,\quad b_2 = n - d,\quad e = d - t,\quad
a_1 = d - z,\quad a_2 = t - z,
\]
\[
v_1 = \frac{a_2}{\gcd(a_2, a_1 - a_2)}, \qquad
v_2 = \frac{a_1 - a_2}{\gcd(a_2, a_1 - a_2)}.
\]

Then $m = a_1 v_1 = \mathrm{lcm}(a_1, a_2)$. The dealer Alice considers a vector
\[
\bm{m} = (e_1, e_2, \ldots, e_m) \in \mathbb{F}_q^m,
\]
which corresponds to the secret quantum state
\[
|\bm{m}\rangle = |e_1, e_2, \ldots, e_m\rangle \in (\mathbb{C}^q)^{\otimes m}.
\]

\subsubsection{Pre-encryption of the Quantum Secret Using a Classical $(t,n)$ Threshold Scheme}

Alice randomly selects two classical vectors of length $m$,
\[
\bm{a} = (a^1, \ldots, a^m) \in \mathbb{F}_q^m,\quad
\bm{b} = (b^1, \ldots, b^m) \in \mathbb{F}_q^m,
\]
and denotes by $K = (\bm{a}, \bm{b}) \in \mathbb{F}_q^{2m}$ the classical keys used for pre-encrypting the quantum secret.

Alice then applies a classical $(t,n)$ threshold sharing scheme to split the key $K$. Specifically, for each $l \in \{1,2,\ldots,m\}$, she constructs two polynomials of degree at most $t-1$,
\begin{equation}
f_l(x) = a^l + \sum_{r=1}^{t-1} \alpha_{l,r} x^r, \qquad
g_l(x) = b^l + \sum_{r=1}^{t-1} \beta_{l,r} x^r,
\tag{22}
\label{eq.22}
\end{equation}
where the coefficients $\alpha_{l,r}, \beta_{l,r} \in \mathbb{F}_q$ are chosen independently and uniformly at random by Alice.

For each participant $\mathrm{Bob}_i$, define the corresponding classical share as
\[
K_i = \bigl(f_1(\xi_i), \ldots, f_m(\xi_i),\; g_1(\xi_i), \ldots, g_m(\xi_i)\bigr).
\]

Using the classical key $K = (\bm{a}, \bm{b})$, define the quantum pre-encryption operator by
$U_K = \bigotimes_{l=1}^{m} X_l^{a^l} Z_l^{b^l},$
where $X_l$ and $Z_l$ denote the generalized Pauli $X$- and $Z$-operators acting on the $l$-th qudit, respectively.

Alice applies this operator to the original quantum secret and obtains the encrypted quantum secret state
\[
|S\rangle = U_K |bm{m}\rangle = \bigotimes_{l=1}^{m} X_l^{a^l} Z_l^{b^l} |m\rangle.
\]

The encrypted quantum secret state $|S\rangle \in \mathbb{F}_q^m$ is then partitioned into $v_1$ blocks of length $a_1$:
\[
\bm{s}_1 = (s_1, \ldots, s_{a_1}), \quad
\bm{s}_2 = (s_{a_1+1}, \ldots, s_{2a_1}), \ldots, \
\bm{s}_{v_1}= (s_{a_1(v_1-1)+1}, \ldots, s_{a_1v_1}).
\]

These column vectors are stacked to form the matrix
\[
S = (\bm{s}_1, \bm{s}_2, \ldots, \bm{s}_{v_1}) \in \mathbb{F}_q^{a_1 \times v_1}.
\]
\subsubsection{Construction of the Random Matrices and  Encoding }

Let $r_{ij} \in \mathbb{F}_q$, and introduce the random matrices
\[
R_{1,1} = (r_{ij})_{b_2 \times v_1}, \quad
R_{1,2} = (r_{ij})_{e \times v_1}, \quad
R_2 = (r_{ij})_{b_1 \times v_2}.
\]
We define the matrix $Y \in \mathbb{F}_q^{d \times (v_1+v_2)}$ by
\begin{equation}
Y=
\left[
\begin{array}{c|c}
S & \begin{matrix}0\\D\end{matrix} \\
\begin{matrix}R_{1,1}\\R_{1,2}\end{matrix} & R_2
\end{array}
\right]
\tag{23}
\label{eq.23}
\end{equation}

Choose $n$ distinct nonzero elements $x_1, x_2, \ldots, x_n \in \mathbb{F}_q$, and define the matrix $G_{B_0}$ over $\mathbb{F}_q$ accordingly. 

\[
G_{B_0}
=
\begin{bmatrix}
G_{A_1} \\
G_{B_1}
\end{bmatrix}
=
\begin{bmatrix}
G_{A_1/A_2} \\
G_{A_2} \\
G_{B_2} \\
G_E
\end{bmatrix}
=
\begin{bmatrix}
1 & 1 & \cdots & 1 \\
x_1 & x_2 & \cdots & x_n \\
x_1^2 & x_2^2 & \cdots & x_n^2 \\
\vdots & \vdots &  & \vdots \\
x_1^{a_1-a_2-1} & x_2^{a_1-a_2-1} & \cdots & x_n^{a_1-a_2-1} \\
x_1^{a_1-a_2} & x_2^{a_1-a_2} & \cdots & x_n^{a_1-a_2} \\
\vdots & \vdots &  & \vdots \\
x_1^{a_1-1} & x_2^{a_1-1} & \cdots & x_n^{a_1-1} \\
\vdots & \vdots &  & \vdots \\
x_1^{a_1+b_2-1} & x_2^{a_1+b_2-1} & \cdots & x_n^{a_1+b_2-1} \\
\vdots & \vdots &  & \vdots \\
x_1^{b_0-1} & x_2^{b_0-1} & \cdots & x_n^{b_0-1}
\end{bmatrix}.
\]

The matrix $G_{B_0}$ is publicly known to the combiner and all participants.Consider the matrix
$C = G_{B_0} Y.$
Then, for a basis state $(s_1,s_2,\ldots,s_m) \in \mathbb{F}_q^m$, the encoding rule is given by $\mathcal{E}$ as
\begin{equation}
\mathcal{E}: |s_1 s_2 \cdots s_m\rangle \mapsto
|S\rangle =
\sum_{R_{1,1}, R_{1,2}, R_2}
\bigotimes_{i=1}^{n}
|c_{i1} c_{i2} \cdots c_{i(v_1+v_2)}\rangle,
\tag{24}
\label{eq.24}
\end{equation}
where the normalization factor is omitted.

When the encoded state $|\psi\rangle$ is distributed, the particles associated with each participant are organized according to the participant label. For $\mathrm{Bob}_i$, we denote by $P_{(i,j)}$ the particle located at position
$(i-1)(v_1+v_2)+j$
in the encoded superposition state. Thus, $P_{(i,j)}$ corresponds to the $j$-th information particle assigned to $\mathrm{Bob}_i$. Under this indexing convention, the information particles contained in the encoded secret state are arranged into the following sequences:

\[
C_1 = \{P_{11}, P_{12}, \ldots, P_{1(v_1+v_2)}\},
\]
\[
C_2 = \{P_{21}, P_{22}, \ldots, P_{2(v_1+v_2)}\},
\]
\[
\vdots
\]
\[
C_n = \{P_{n1}, P_{n2}, \ldots, P_{n(v_1+v_2)}\}.
\]

\subsection{Distribution Phase}

Suppose that Alice plans to distribute a total of $B$ code blocks, among which 
$B_{\mathrm{sec}}$ blocks are secret code blocks carrying the quantum secret, and 
$B_{\mathrm{test}}$ blocks are test code blocks used for eavesdropping detection. Thus,$B = B_{\mathrm{sec}} + B_{\mathrm{test}}.$

\subsubsection{Selection of Test Code Blocks and Preparation of Logical States}

Alice randomly selects a test-block index set
 $\tau \subset \{1,2,\ldots,B\}, 
    \qquad |\tau| = B_{\mathrm{test}} $.
For each $b\in\tau$, Alice randomly assigns a test type
 $ \mathrm{Type}(b)\in \{Z\text{-test}, X\text{-test}\}.$

For each code block $b=1,\ldots,B$, Alice first prepares a logical state
$|\phi_b\rangle$ in the logical space. Specifically, if $b\in\tau$ and
$\mathrm{Type}(b)=Z\text{-test}$, Alice prepares the logical zero state
$|0_L\rangle$, which is used to detect $X$-type disturbances. If $b\in\tau$
and $\mathrm{Type}(b)=X\text{-test}$, Alice prepares the logical plus state
$|+_L\rangle$, which is used to detect $Z$-type disturbances. If $b\notin\tau$,
then the block is a secret code block, and Alice encodes the quantum secret into
the logical state $ |\phi_b\rangle = |\psi_L\rangle .$
\begin{figure}[htbp]
    \centering
    \includegraphics[width=0.8\textwidth]{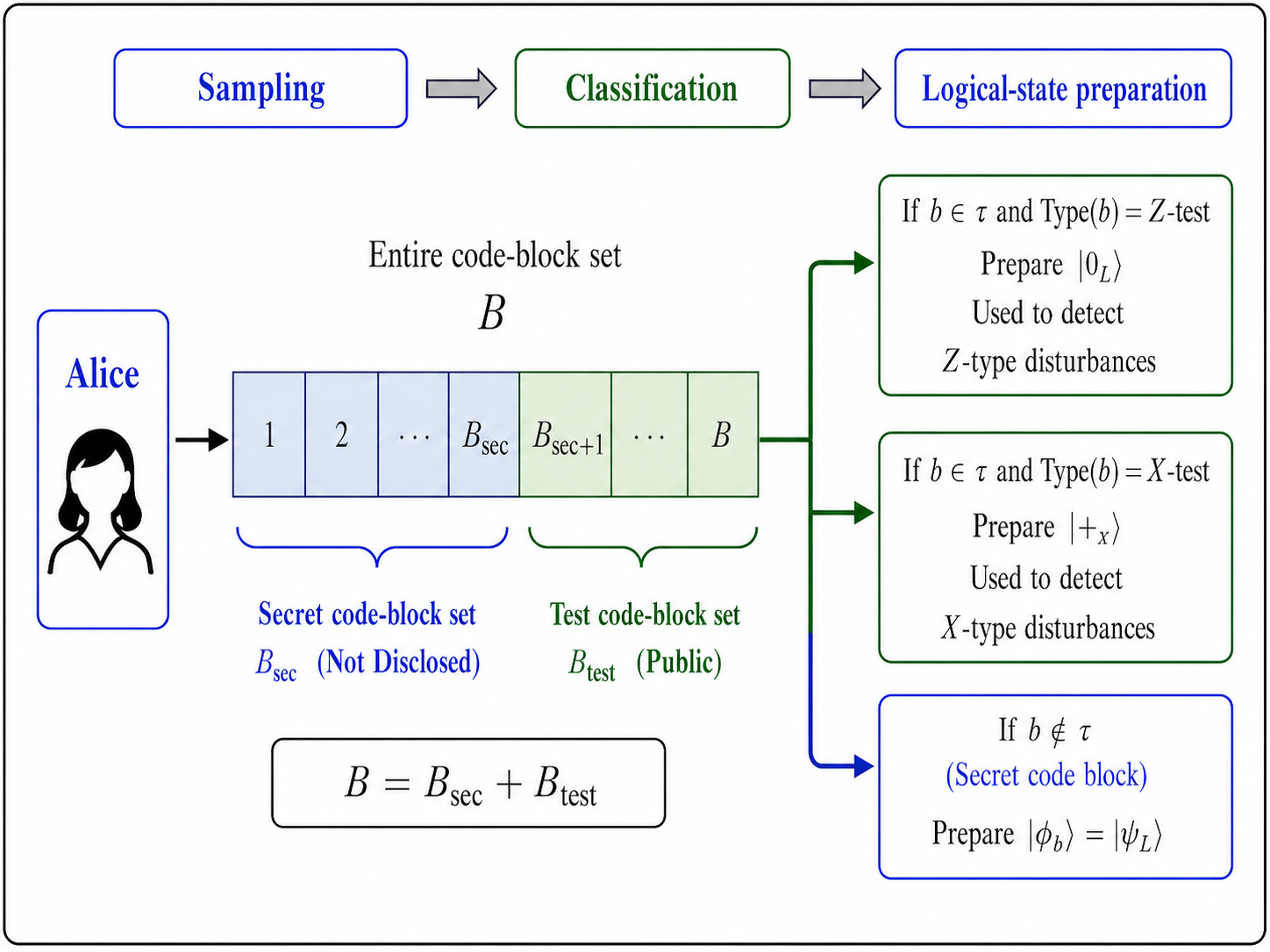}
    \caption{Selection of test code blocks and preparation of logical states}
    \label{fig2}
\end{figure}

\subsubsection{ECSS Encoding and Share Distribution}

After the logical states are determined, Alice applies the ECSS encoding operation
to each code block. The logical state is mapped to an $n$-partite physical quantum
code-block state
\[
    |\Phi_b\rangle = U_{\mathrm{ECSS}}|\phi_b\rangle .
\]
Let $\Phi_b^{(j)}$ denote the $j$-th physical qudit of $|\Phi_b\rangle$, where
$j=1,\ldots,n$. This qudit is distributed to participant $\mathrm{Bob}_j$.

For each participant $\mathrm{Bob}_j$, Alice sends the quantum share sequence
\[
    \left(\Phi_1^{(j)},\Phi_2^{(j)},\ldots,\Phi_B^{(j)}\right)
\]
through a quantum channel. In addition to the quantum share sequence, Alice also
sends the corresponding classical threshold share $K_j$ to $\mathrm{Bob}_j$
through an authenticated classical secure channel, where $j=1,\ldots,n$. After
receiving the shares, each participant confirms receipt through the authenticated
classical channel.
\begin{figure}[htbp]
    \centering
    \includegraphics[width=0.8\textwidth]{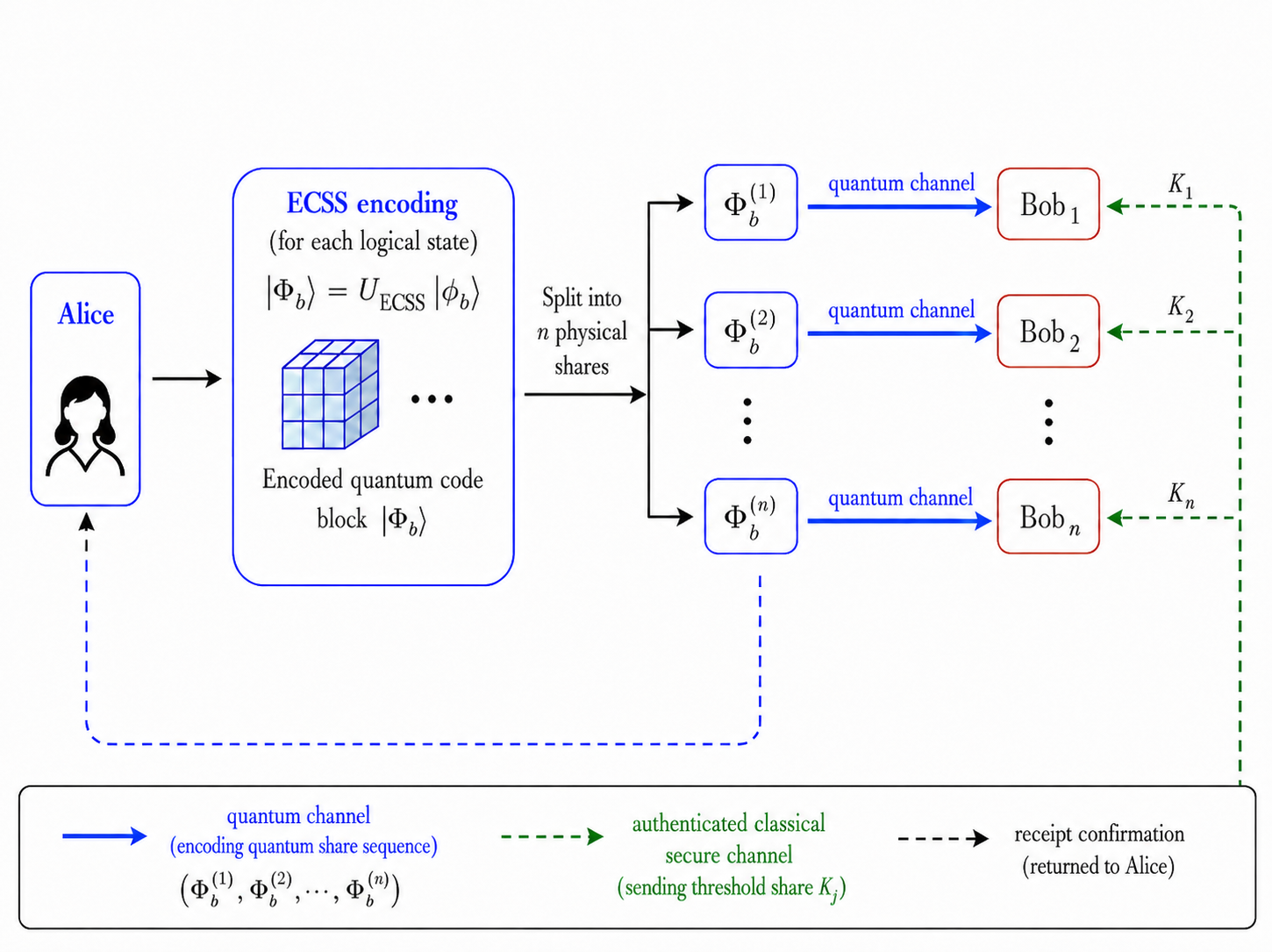}
    \caption{ECSS encoding and share distribution}
    \label{fig:ecss-share-distribution}
\end{figure}
\subsubsection{Local Measurements and Detection Decision}

After all receipt confirmations are completed, Alice announces the test-block
index set $\tau$ and the corresponding test types
\[
    \{\mathrm{Type}(b): b\in\tau\}
\]
through the authenticated classical channel. The information of the secret code
blocks is not disclosed.

For a $Z$-test block, participant $\mathrm{Bob}_j$ measures the qudit
$\Phi_b^{(j)}$ in the computational basis and obtains the measurement outcome
\[
    x_j(b)\in\mathbb{Z}_q .
\]
Then $\mathrm{Bob}_j$ computes
\[
    u_{rj}^{(Z)}(b) \equiv h_{rj}^{(Z)}x_j(b) \pmod q .
\]

For an $X$-test block, participant $\mathrm{Bob}_j$ measures the qudit
$\Phi_b^{(j)}$ in the Fourier basis and obtains the measurement outcome
\[
    k_j(b)\in\mathbb{Z}_q .
\]
Then $\mathrm{Bob}_j$ computes
\[
    u_{rj}^{(X)}(b) \equiv -h_{rj}^{(X)}k_j(b) \pmod q .
\]

Each participant sends $u_{rj}^{(Z)}(b)$ or $u_{rj}^{(X)}(b)$ to Alice through the
authenticated classical channel. Alice then sums the local exponents modulo $q$
over all participants and obtains the $r$-th $Z$-syndrome component and the $r$-th
$X$-syndrome component:
\[
    s_{Z,r}(b) \equiv \sum_{j=1}^{n} u_{rj}^{(Z)}(b) \pmod q,
    \qquad
    s_{X,r}(b) \equiv \sum_{j=1}^{n} u_{rj}^{(X)}(b) \pmod q .
\]

For a test block $b$, if there exists an index $r$ such that
\[
    s_{Z,r}(b)\neq 0
    \quad \text{or} \quad
    s_{X,r}(b)\neq 0 ,
\]
then the test block is judged to fail the stabilizer detection, indicating the
presence of abnormal disturbance or a potential attack. Otherwise, the test block
is regarded as passing the stabilizer detection.
\begin{figure}[htbp]
    \centering
    \includegraphics[width=0.8\textwidth]{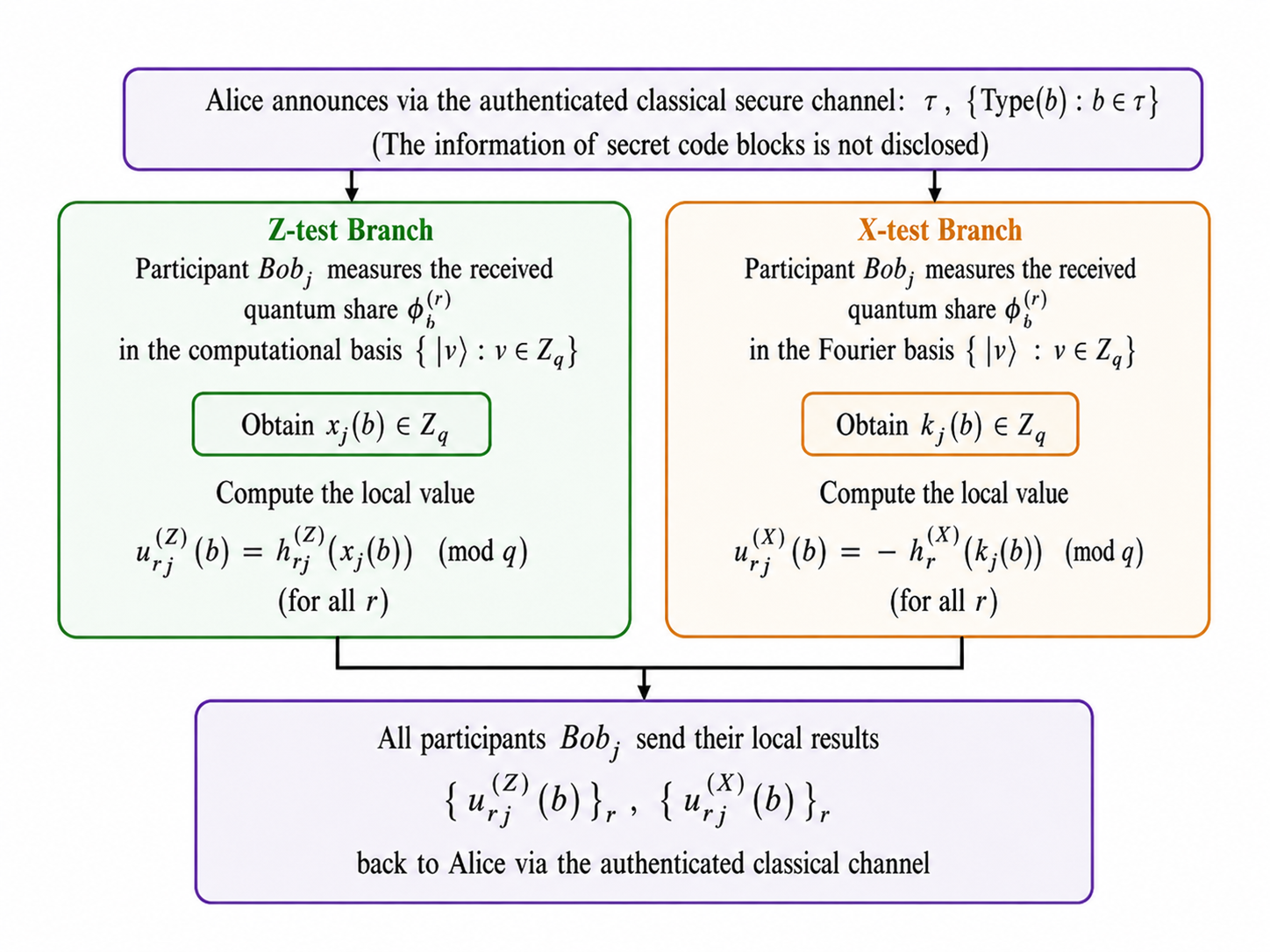}
    \caption{Local measurements by participants}
    \label{fig4}
\end{figure}

\subsubsection{Selection of Test Code Blocks and Preparation of Logical States
}
Alice randomly selects a test-block index set $\tau \subset \{1,2,\ldots,B\}$, with $|\tau| = B_{\mathrm{test}}$, and for each $b \in \tau$, randomly specifies its test type
$\mathrm{Type}(b) \in \{Z\text{-test}, X\text{-test}\}.$

For each code block $b=1,\ldots,B$, Alice first prepares a logical state $|\phi_b\rangle$ in the logical space. If $b \in \tau$ and $\mathrm{Type}(b)=Z$-test, she prepares the logical zero state $|0_L\rangle$ to detect $X$-type disturbances. If $b \in \tau$ and $\mathrm{Type}(b)=X$-test, she prepares the logical plus state $|+_L\rangle$ to detect $Z$-type disturbances. If $b \notin \tau$ (i.e., $b$ is a secret block), Alice encodes the secret state as the logical state $|\phi_b\rangle = |\psi_L\rangle$.
\begin{figure}[htbp]
    \centering
    \includegraphics[width=0.8\textwidth]{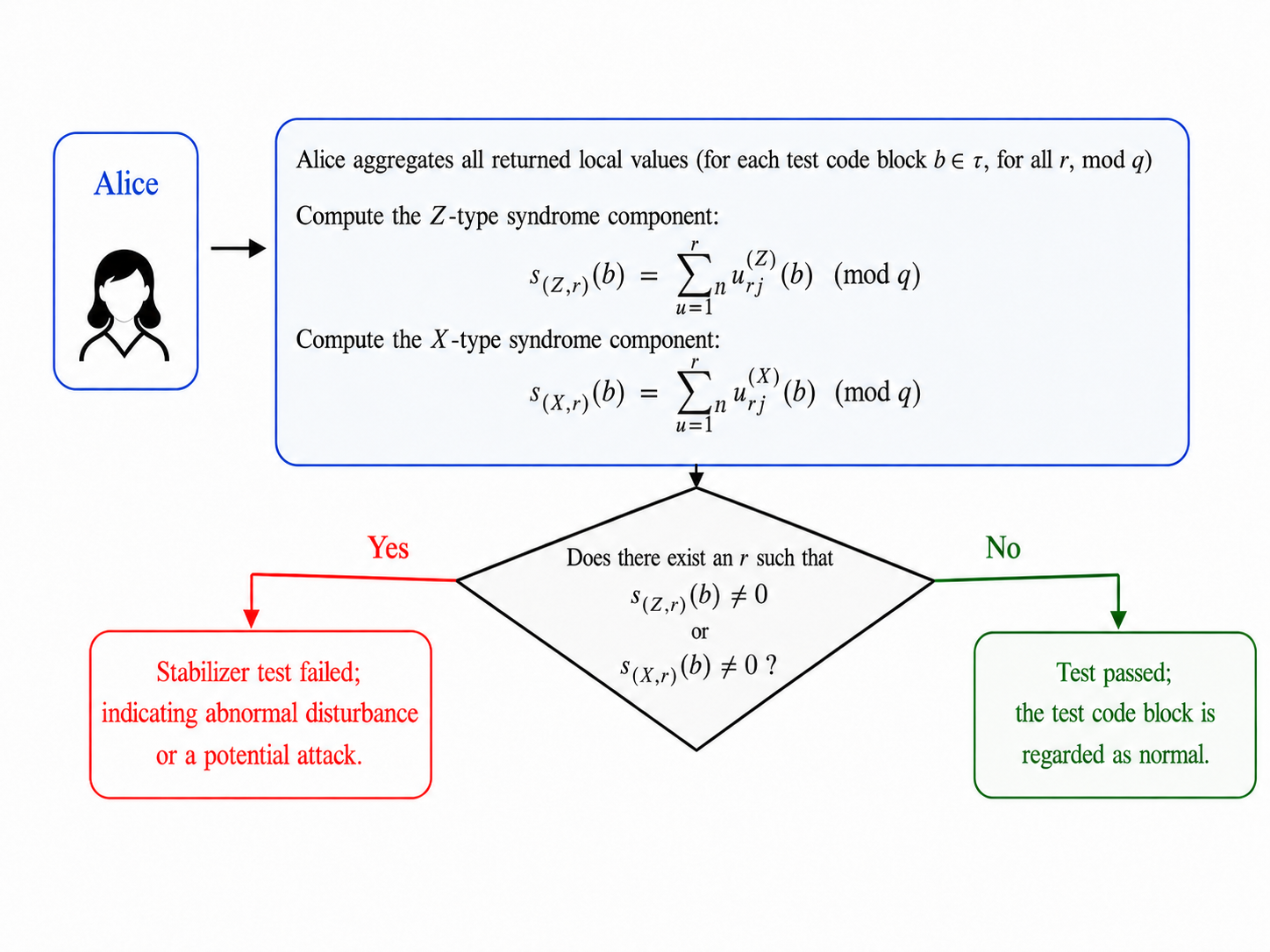}
    \caption{Detection decision based on the syndrome vector}
    \label{fig:syndrome-vector-detection}
\end{figure}

\subsection{Eavesdropping Detection Phase}

Let $N_{\mathrm{tot}}$ denote the total number of stabilizer checks performed over all test blocks, and let $N_{\mathrm{fail}}$ denote the number of failed checks. Define the empirical error rate by
\[
\hat{e} = \frac{N_{\mathrm{fail}}}{N_{\mathrm{tot}}}.
\]

Alice sets a threshold $e_{\mathrm{th}}$ in advance. If $\hat{e} > e_{\mathrm{th}}$, the protocol is aborted and the presence of eavesdropping or severe disturbance is declared. Otherwise, if $\hat{e} \le e_{\mathrm{th}}$, the channel is regarded as acceptable, and the secret blocks are retained for the subsequent secret reconstruction phase.

\subsection{Reconstruction Phase}

\subsubsection{Channel Checking for Transmission from Participants to the Combiner}

When participant $\mathrm{Bob}_j$ is required to take part in secret reconstruction, 
he should transmit to the combiner Peter the qudits specified by the protocol 
that are needed for Peter's reconstruction operation. To detect possible 
eavesdropping or strong disturbance on the quantum channel 
$\mathrm{Bob}_j\rightarrow \mathrm{Peter}$, the sender inserts decoy particles into 
the transmission sequence and performs a sampling test. The detailed procedure is 
as follows.

Let $M_j$ denote the set of qudits that $\mathrm{Bob}_j$ needs to send for secret 
reconstruction. Participant $\mathrm{Bob}_j$ randomly prepares $N$ decoy qudits and 
inserts them into $M_j$ at random positions, thereby forming the final transmission 
sequence $L_j$. He then sends $L_j$ to Peter through the quantum channel. After 
Peter receives the sequence, he returns an acknowledgement through an authenticated 
classical channel. Upon receiving this acknowledgement, $\mathrm{Bob}_j$ announces, 
through the authenticated classical channel, the positions of the decoy qudits in 
$L_j$ and the preparation basis associated with each decoy qudit.

Peter measures the decoy qudits according to the announced bases and sends the 
measurement outcomes back to $\mathrm{Bob}_j$. Participant $\mathrm{Bob}_j$ compares 
Peter's reported outcomes with the originally prepared values and computes the 
decoy error rate $\hat{e}$. If $\hat{e}>e_{\mathrm{th}}$ ,then the channel is judged to be insecure and the protocol is aborted. Otherwise, 
the transmission is regarded as acceptable. Peter keeps the information qudits in 
the received sequence and proceeds to the subsequent secret reconstruction phase.

\subsubsection{Secret Reconstruction}

Suppose that the combiner Peter selects $k$ participants from the $n$ participants to participate in secret recovery, where $k \ge t$. Let the selected participant set be
$J = \{i_1, i_2, \ldots, i_k\} \subseteq \{1,2,\ldots,n\},$
and let its complement be denoted by
$\bar{J} = \{1,2,\ldots,n\} \setminus J.$

From the encoding relation $C = VY$, it follows that the encoded secret state can be expressed as a uniform superposition over the random variables. Omitting the normalization factor, the global state is given by
\begin{equation}
\sum_{R_{1,1}, R_{1,2}, R_2}
\left|
V_{A_1}^T S + V_{B_2}^T R_{1,1} + V_E^T R_{1,2}
\right\rangle_{[n],1}
\left|
V_{A_2}^T D_1 + V_{B_1}^T R_2
\right\rangle_{[n],2}.
\tag{25}
\label{eq.25}
\end{equation}

By rearranging the particles in the encoded secret state in Eq.~(\ref{eq.25}), we obtain
\begin{align}
\sum_{R_{1,1}, R_{1,2}, R_2}
&\left|
V_{A_1}^{(J)T} S + V_{B_2}^{(J)T} R_{1,1} + V_E^{(J)T} R_{1,2}
\right\rangle_{J,1}
\left|
V_{A_2}^{(J)T} D_1 + V_{B_1}^{(J)T} R_2
\right\rangle_{J,2}
\nonumber\\
&\otimes
\left|
V_{A_1}^{(\bar{J})T} S + V_{B_2}^{(\bar{J})T} R_{1,1} + V_E^{(\bar{J})T} R_{1,2}
\right\rangle_{\bar{J},1}
\left|
V_{A_2}^{(\bar{J})T} D_1 + V_{B_1}^{(\bar{J})T} R_2
\right\rangle_{\bar{J},2}.
\tag{26}
\label{eq.26}
\end{align}

Here, $V_{A_1}^T, V_{A_2}^T, V_{B_1}^T, V_{B_2}^T, V_E^T$ denote the corresponding block submatrices of $V$.
\medskip
\noindent\textbf{Case 1:} $(t \le k < d)$.

Let $J=\{i_1,i_2,\ldots,i_k\}\subseteq \{1,2,\ldots,n\}$ denote the set of the
selected $k$ participants, and let 
$\overline{J}=\{i_{k+1},i_{k+2},\ldots,i_n\}$ denote the complement of $J$.

Since $\left[V_{A_2}^{T},V_{B_1}^{T}\right]_{J}$ is a $k\times t$ matrix with
$k\geq t$, the combiner can select, from the second-layer qudits corresponding
to the participants in $J$, the components associated with $t$ participants so
as to form an invertible $t\times t$ submatrix, denoted by $V_1$. Then, by
applying the unitary operation $U_{V_1^{-1}}$ to the $k v_2$ qudits in Eq.~(\ref{eq.26})
and subsequently discarding the remaining qudits in the second layer, we obtain
\begin{align}
\sum_{R_{1,1},R_{1,2},R_2}
&\left|
V_{A_1}^{(J)T}S
+V_{B_2}^{(J)T}R_{1,1}
+V_E^{(J)T}R_{1,2}
\right\rangle_{J,1}
|D_1\rangle_{J,2}
|R_2\rangle_{J,2}
\nonumber\\
&\otimes
\left|
V_{A_1}^{(\bar{J})T}S
+V_{B_2}^{(\bar{J})T}R_{1,1}
+V_E^{(\bar{J})T}R_{1,2}
\right\rangle_{\bar{J},1}
\left|
V_{A_2}^{(\bar{J})T}D_1
+V_{B_1}^{(\bar{J})T}R_2
\right\rangle_{\bar{J},2}
\tag{27}
\label{eq.27}
\end{align}

By applying the known permutation to $|D_1\rangle_{J,2}$, one can recover
$|R_{1,2}\rangle_{J,2}$, and hence obtain
\begin{align}
\sum_{R_{1,1},R_{1,2},R_2}
&\left|
V_{A_1}^{(J)T}S
+V_{B_2}^{(J)T}R_{1,1}
+V_E^{(J)T}R_{1,2}
\right\rangle_{J,1}
|R_{1,2}\rangle_{J,2}
|R_2\rangle_{J,2}
\nonumber\\
&\otimes
\left|
V_{A_1}^{(\bar{J})T}S
+V_{B_2}^{(\bar{J})T}R_{1,1}
+V_E^{(\bar{J})T}R_{1,2}
\right\rangle_{\bar{J},1}
\left|
V_{A_2}^{(\bar{J})T}D_1
+V_{B_1}^{(\bar{J})T}R_2
\right\rangle_{\bar{J},2}.
\tag{28}
\label{eq.28}
\end{align}

Define the block-diagonal matrix
$W_1 =
\begin{bmatrix}
V_E^{(J)T} & 0 \\
0 & I
\end{bmatrix}.$
Clearly, $W_1$ is invertible. The combiner then applies the unitary operation
$U_{W_1}$ to the particles $|R_{1,2},R_2\rangle_{J,2}$, and subsequently applies
$\mathrm{ADD}^{\dagger}$ gates to the corresponding registers, thereby
eliminating the term $V_E^{(J)T}R_{1,2}$ from the first layer. This yields
\begin{align}
\sum_{R_{1,1},R_{1,2},R_2}
&\left|
V_{A_1}^{(J)T}S
+V_{B_2}^{(J)T}R_{1,1}
\right\rangle_{J,1}
\left|
V_E^{(\bar{J})T}R_{1,2}
\right\rangle_{J,2}
|R_2\rangle_{J,2}
\nonumber\\
&\otimes
\left|
V_{A_1}^{(\bar{J})T}S
+V_{B_2}^{(\bar{J})T}R_{1,1}
+V_E^{(\bar{J})T}R_{1,2}
\right\rangle_{\bar{J},1}
\left|
V_{A_2}^{(\bar{J})T}D_1
+V_{B_1}^{(\bar{J})T}R_2
\right\rangle_{\bar{J},2}.
\tag{29}
\label{eq.29}
\end{align}

Moreover,
$\bigl[V_{A_1}^{T},V_{B_2}^{T}\bigr]^J$
is a $k \times t$ matrix, where $k \ge t$. By its structure, there exists an
invertible $t \times t$ submatrix, denoted by $V_2$. The combiner then selects
the first-layer qudits corresponding to $V_2$ from the set $J$, and applies the
unitary operation $U_{V_2^{-1}}$ to Eq.~(\ref{eq.29}), thereby obtaining
\begin{align*}
    |S\rangle
\sum_{R_{1,1},R_{1,2},R_2}
&|R_{1,1}\rangle_{J,1}
\left|
V_E^{(\bar{J})T}R_{1,2}
\right\rangle_{J,2}
|R_2\rangle_{J,2}
\nonumber
\left|
V_{A_1}^{(\bar{J})T}S
+V_{B_2}^{(\bar{J})T}R_{1,1}
+V_E^{(\bar{J})T}R_{1,2}
\right\rangle_{\bar{J},1}\\
&\left|
V_{A_2}^{(\bar{J})T}D_1
+V_{B_1}^{(\bar{J})T}R_2
\right\rangle_{\bar{J},2}.
\tag{30}
\label{eq.30}
\end{align*}

Next, the combiner applies the unitary operation $U_{W_1^{-1}}$ to the particles,
which transforms Eq.~(\ref{eq.30})) into
\begin{align}
|S\rangle
\sum_{R_{1,1},R_{1,2},R_2}
&|R_{1,1}\rangle_{J,1}
|R_{1,2}\rangle_{J,2}
|R_2\rangle_{J,2}
\nonumber
\otimes
\left|
V_{A_1}^{(\bar{J})T}S
+V_{B_2}^{(\bar{J})T}R_{1,1}
+V_E^{(\bar{J})T}R_{1,2}
\right\rangle_{\bar{J},1}\\
&\left|
V_{A_2}^{(\bar{J})T}D_1
+V_{B_1}^{(\bar{J})T}R_2
\right\rangle_{\bar{J},2}.
\tag{31}
\label{eq.31}
\end{align}
To show that the secret is completely disentangled from the remaining system,
so that the secret register $|S\rangle$ is in a tensor-product form with the
other registers, define
$W_2 =
\begin{bmatrix}
I & 0 \\
V^{(J)T} & I
\end{bmatrix},$
where $W_2$ is invertible. The combiner then applies the unitary operation
$U_{W_2}$ to the particles $|S, R_{1,1}, R_{1,2}\rangle$. As a result, the state
in Eq.~(\ref{eq.31}) becomes
\begin{align}
|S\rangle \sum_{R_{1,2},R_2}
\Bigg(
\sum_{R_{1,1}}
&\left|
V^{(J)T}(S,R_{1,1},R_{1,2})
\right\rangle_{J,1}
\left|
V^{(J)T}(S,R_{1,1},R_{1,2})
\right\rangle_{J,1}
\Bigg)
\nonumber\\
&\otimes
|R_{1,2}\rangle_{J,2}
|R_2\rangle_{J,2}
\left|
V_{A_2}^{(\bar{J})T}D_1 + V_{B_1}^{(\bar{J})T}R_2
\right\rangle_{\bar{J},2}.
\tag{32}
\label{eq.32}
\end{align}

Let $f \in \mathbb{F}_p^{\,n-k}$. Then the inner summation in Eq.~~(\ref{eq.32}) can be
rewritten as
\[
\sum_{f \in \mathbb{F}_p^{\,n-k}} |f\rangle |f\rangle.
\]
Therefore, the encoded secret state can be rewritten as
\begin{equation}
|S\rangle
\sum_{f \in \mathbb{F}_p^{\,n-k}} |f\rangle |f\rangle
\sum_{R_{1,2},R_2}
|R_{1,2}\rangle_{J,2}
|R_2\rangle_{J,2}
\left|
V_{A_2}^{(\bar{J})T}D_1 + V_{B_1}^{(\bar{J})T}R_2
\right\rangle_{\bar{J},2}.
\tag{33}
\label{eq.33}
\end{equation}

Hence, the secret has been completely disentangled from the rest of the system.
This conclusion also holds for an arbitrary superposition secret. Therefore, the
combiner can preserve the coherence of the secret state and correctly recover
the quantum secret.

\medskip
\noindent\textbf{Case 2:} $(k \ge d)$.

When the number of participants involved in recovery satisfies $k \ge d$, the
combiner Peter need not collect the second-layer qudits. Instead, it suffices for
the recovering participants $\mathrm{Bob}_{i_k}$ to send only the first-layer
secret information to the combiner, thereby reducing the quantum communication
cost in the reconstruction phase. The specific recovery procedure is as follows.

Ignoring the normalization factor, the first-layer encoded state can be written
as
\[
\sum_{R_{1,1},R_{1,2}}
\left|
V^T(S,R_{1,1},R_{1,2})
\right\rangle_{[n],1}.
\]

Let
$V_3 = \bigl[V_{A_1}^T, V_{B_2}^T, V_E^T\bigr],$
where $V_3$ is an $n \times d$ matrix. Then
\begin{equation}
\sum_{R_{1,1},R_{1,2},R_2}
\left|
V_3^{(J)}(S,R_{1,1},R_{1,2})
\right\rangle_{J,1}
\left|
V_{A_1}^{(\bar{J})T}S + V_{B_2}^{(\bar{J})T}R_{1,1}
+ V_E^{(\bar{J})T}R_{1,2}
\right\rangle_{\bar{J},1}.
\tag{34}
\label{eq.34}
\end{equation}

The combiner arbitrarily selects the first-layer qudits corresponding to any
$d$ participants among those involved in secret recovery, and lets
\[
(V_3^{(J)})^T =
\bigl[(V_3^{(J)})_{A_1}^T, (V_3^{(J)})_{B_2}^T, (V_3^{(J)})_E^T\bigr].
\]
By construction, $(V_3^{(J)})^T$ is a $d \times d$ Vandermonde matrix and is
therefore invertible. The combiner then applies the unitary operation
$U_{V_3^{-1}}$ to the $k v_1$ first-layer qudits in Eq.~(\ref{eq.34}), obtaining
\begin{equation}
\sum_{R_{1,1},R_{1,2},R_2}
|S\rangle
|R_{1,1}\rangle_{J,1}
|R_{1,2}\rangle_{J,1}
\left|
V_{A_1}^{(\bar{J})T}S + V_{B_2}^{(\bar{J})T}R_{1,1}
+ V_E^{(\bar{J})T}R_{1,2}
\right\rangle_{\bar{J},1}.
\tag{35}
\label{eq.35}
\end{equation}

Now define
$W_3 =
\begin{bmatrix}
I & 0 \\
V_3^{(\bar{J})T} & I
\end{bmatrix}.$
The combiner applies the unitary operation $U_{W_3}$ to
$|S,R_{1,1},R_{1,2}\rangle_{J,1}$. Then the state in Eq.~(\ref{eq.35}) becomes
\begin{equation}
\sum_{R_{1,1},R_{1,2},R_2}
|S\rangle
\left|
V_3^{(J)}(S,R_{1,1},R_{1,2})
\right\rangle_{J,1}
\left|
V_{A_1}^{(\bar{J})T}S + V_{B_2}^{(\bar{J})T}R_{1,1}
+ V_E^{(\bar{J})T}R_{1,2}
\right\rangle_{\bar{J},1}.
\tag{36}
\label{eq.36}
\end{equation}

Let $f \in \mathbb{F}_p^{\,n-k}$. Then the term in Eq.~~(\ref{eq.36}) can be rewritten as
$\sum_{f \in \mathbb{F}_p^{\,n-k}} |f\rangle |f\rangle.$
Accordingly, the encoded secret state can be rewritten as
\begin{equation}
|S\rangle
\sum_{f \in \mathbb{F}_p^{\,n-k}} |f\rangle |f\rangle
\sum_{R_{1,2}}
|R_{1,2}\rangle_{J,1}.
\tag{37}
\label{eq.37}
\end{equation}

The recovered state $|S\rangle$ at this stage is precisely the classically
encrypted quantum secret state $|\tilde{\psi}\rangle$.

The authorized set participating in the recovery first uses its classical shares
to reconstruct the quantum pre-encryption key $K=(\bm{a},\bm{b})$. Peter then
reconstructs the encrypted quantum secret state $|S\rangle$ according to the
above recovery algorithm, and finally applies the inverse transformation
$U_K^{-1} = \bigotimes_{l=1}^{m} X_l^{-a^l} Z_l^{-b^l}$
to recover the original quantum secret state $|m\rangle$.
\section{Construction of the $(4,6,5;2)_7$ QSS Scheme}

Section~4 presented the general encoding and reconstruction framework of the
$\left((t,n,d;z)\right)_q$ scheme. To make the invertibility of the Vandermonde
submatrices involved in the reconstruction process, the elimination of random
offsets, and the disentanglement between the secret and auxiliary systems more
transparent, this section provides a complete example under the parameter setting
$\left((t,n,d;z)\right)_q=\left((4,6,5;2)\right)_7 .$
In this section, we only discuss the encryption and reconstruction phases. The
preparation, distribution, and detection phases are the same as those described
above and are therefore omitted here.

\subsection{Intermediate-Set Leakage Analysis}

Let the secret vector be
$\bm{m}=(s_1,s_2,s_3,s_4,s_5,s_6),$
where $a_1=3, a_2=2, b_0=5, b_1=2, b_2=1, e=1,v_1=2, v_2=1.$
Under this parameter setting, the quantum secret to be shared consists of six
qudits, and its secret vector is denoted by $\bm{m}$. Each participant finally
holds a quantum share consisting of three qudits. Take
\[
V^T =
\begin{bmatrix}
1 & 1 & 1 & 1 & 1 \\
1 & 2 & 4 & 1 & 2 \\
1 & 3 & 2 & 6 & 4 \\
1 & 4 & 2 & 1 & 4 \\
1 & 5 & 4 & 6 & 2 \\
1 & 6 & 1 & 6 & 1
\end{bmatrix},
\qquad
Y =
\begin{bmatrix}
s_1 & s_2 & 0 \\
s_3 & s_4 & r_3 \\
s_5 & s_6 & r_4 \\
r_1 & r_2 & r_5 \\
r_3 & r_4 & r_6
\end{bmatrix}.
\]

According to the encoding rule, we have
\[
|m\rangle \mapsto
\sum_{r \in \mathbb{F}_7^6}
\bigotimes_{i=1}^{6} |c_{i1} c_{i2} c_{i3}\rangle.
\]

In order to facilitate the analysis of the linear relationship
between the secret variables and the random variables contained in each
participant's share, the above encoding can also be equivalently written in
matrix form as
\[
\ket{\bm{m}}
\longmapsto
\sum_{\bm{r}\in \mathbb{F}_7^{6}}
\ket{V^{T}Y_1}\otimes \ket{V^{T}Y_2},
\]
where
\[
Y_1=
\begin{bmatrix}
s_1 & s_2\\
s_3 & s_4\\
s_5 & s_6\\
r_1 & r_2\\
r_3 & r_4
\end{bmatrix},
\qquad
Y_2=
\begin{bmatrix}
0\\
r_3\\
r_4\\
r_5\\
r_6
\end{bmatrix}.
\]
Here, $V^{T}Y_1$ corresponds to the first-layer registers of the quantum shares
held by the participants, while $V^{T}Y_2$ corresponds to the second-layer
registers of the quantum shares held by the participants.

Now consider the participant set $\{P_1,P_2,P_4\}$. In this case,
$x_1=1$, $x_2=2$, and $x_4=4$. For $l=1,2$, we have
$7 \mid x_1+x_2+x_4$
and
$7 \mid x_1^2+x_2^2+x_4^2+x_1x_4+x_1x_2+x_2x_4 .$
Therefore, by Theorem~\ref{thm2}, the joint quantum share held by this participant set
can be processed through reversible linear elimination, which reveals the qudit
information associated with $r_4$. By using the generalized $\mathrm{ADD}$ gate
and its inverse to eliminate $r_4$, we obtain the first-layer register.
\[
\sum_{r \in \mathbb{F}_7^4}
\left|
\begin{bmatrix}
1 & 1 & 1 & 1 & 0 \\
1 & 2 & 4 & 1 & 0 \\
1 & 4 & 2 & 1 & 0
\end{bmatrix}
\begin{bmatrix}
s_2 \\ s_4 \\ s_6 \\ r_2 \\ r_4
\end{bmatrix}
\right\rangle.
\]
Since $7 \mid x_1+x_2+x_4$,
by Theorem~\ref{thm2}, some linear combinations in the corresponding registers no longer
depend on all the random variables, but instead contain nontrivial information
about the secret components $s_4$ and $s_6$. Hence, an intermediate set exists.

This indicates that intermediate-set leakage indeed occurs in the original
Non-PQSS construction. To eliminate this security risk, we introduce a classical
$(4,6)$ threshold sharing layer on top of the original scheme. This classical
layer is used to protect the classical key required for pre-encrypting the
quantum secret, thereby yielding a hybrid perfect quantum secret sharing scheme.

\subsection{Distribution Phase of the $\left((4,6,5;2)\right)_7$ Scheme}
Alice first chooses two classical pre-encryption keys of length $6$:
$\bm{a}=(2,4,1,6,3,5)\in \mathbb{F}_7^6,
\bm{b}=(1,3,2,5,4,6)\in \mathbb{F}_7^6,$
which are used as the $X$-type and $Z$-type keys, respectively, for the
pre-encryption of the quantum secret. Alice then applies
$U_K=\bigotimes_{i=1}^{6} X_i^{a_i}Z_i^{b_i}$
to the original quantum state $\ket{\bm{m}}$, obtaining the classically
pre-encrypted quantum secret state $\ket{\bm{s}}$.

It should be emphasized that $\bm{a}$ and $\bm{b}$ are not disclosed directly.
Instead, they are distributed to the participants through a classical $(4,6)$
threshold secret sharing scheme. In this way, even if some intermediate sets can
obtain partial information about $\ket{\bm{s}}$ from their quantum shares, they
cannot recover the original quantum secret state $\ket{\bm{m}}$ without knowing
the complete pre-encryption key.

More specifically, for $\bm{a}$, Alice constructs polynomials of degree at most
$3$ satisfying $f_j(0)=a_j$:
\[
\begin{aligned}
f_1(x) &= 2+x+2x^2+3x^3, 
&\qquad 
f_2(x) &= 4+2x+x^2+4x^3,\\
f_3(x) &= 1+3x+x^3, 
&\qquad 
f_4(x) &= 6+x+x^2+x^3,\\
f_5(x) &= 3+2x^2+5x^3, 
&\qquad 
f_6(x) &= 5+2x+3x^2+x^3.
\end{aligned}
\]
For $\bm{b}$, Alice constructs polynomials of degree at most $3$ satisfying
\[
\begin{aligned}
g_1(x) &= 1+x+2x^3,
&\qquad
g_2(x) &= 3+x^2+3x^3,\\
g_3(x) &= 2+2x+x^2,
&\qquad
g_4(x) &= 5+3x+2x^2+x^3,\\
g_5(x) &= 4+x+3x^2+2x^3,
&\qquad
g_6(x) &= 6+2x+x^3.
\end{aligned}
\]
Thus, for each participant $P_i$ $(i=1,\ldots,6)$, Alice sends the corresponding
classical threshold share
\[
c_i=\bigl(f_1(i),\ldots,f_6(i),g_1(i),\ldots,g_6(i)\bigr).
\]
These classical shares will be used in the authorized reconstruction phase to
recover the pre-encryption keys $(\bm{a},\bm{b})$.

On the other hand, Alice further encodes the pre-encrypted quantum secret state
$\ket{\bm{s}}$ according to the ECSS encoding rule 
\[
\ket{\bm{s}}
\longmapsto
\sum_{\bm{r}\in \mathbb{F}_7^6}
\bigotimes_{i=1}^{6}
\ket{c_{i1}c_{i2}c_{i3}},
\]
thereby obtaining the complete set of encoded quantum shares. Alice then sends
the quantum share corresponding to participant $P_i$, together with the classical
threshold share $c_i$, to $P_i$.

For clarity in the subsequent reconstruction procedure, the qudit registers
highlighted in bold in the following formulas indicate the quantum shares that
Peter has already collected from the participants and used in the corresponding
unitary operations.
\subsection{Reconstruction Phase of the $\left((4,6,5;2)\right)_7$ Scheme}

\noindent\textbf{Case 1: \(k=6\).} since $V^{T}$ is a $6\times 5$ matrix, the combiner may choose the
first-layer qudits sent by any five participants. Without loss of generality,
let the participating set be
$J=\{1,2,3,4,5\}.$
Then $V_{[5]}$ is a $5\times 5$ Vandermonde matrix. The combiner applies
$U_{V_{[5]}^{-1}}$ to the following registers:
\[
\left|
\bm{V}_{[5]}^{T}
\left(\bm{s}_1,\bm{s}_3,\bm{s}_5,\bm{r}_1,\bm{r}_3\right)
\right\rangle
\left|
\bm{V}_{[5]}^{T}
\left(\bm{s}_2,\bm{s}_4,\bm{s}_6,\bm{r}_2,\bm{r}_4\right)
\right\rangle 
\]
This yields
\[
\ket{\bm{s}_1\bm{s}_2\bm{s}_3\bm{s}_4\bm{s}_5\bm{s}_6}
\sum_{\bm{r}\in \mathbb{F}_7^6}
\ket{\bm{r}_1,\bm{r}_3}\otimes \ket{\bm{r}_2,\bm{r}_4}.
\]

\noindent\textbf{Case 2: \(k=5\).},assume, without loss of generality, that the participating set is
$\{1,2,3,4,5\}$. The qudit state can be written as
\[
\sum_{\bm{r}\in \mathbb{F}_7^6}
\ket{\bm{V}_{[5]}^{T}(\bm{s}_1,\bm{s}_3,\bm{s}_5,\bm{r}_1,\bm{r}_3)}
\ket{\bm{V}_{[5]}^{T}(\bm{s}_2,\bm{s}_4,\bm{s}_6,\bm{r}_2,\bm{r}_4)}
\ket{V_{[6,6]}^{T}(s_1,\bm{s}_3,s_5,s_1,s_3)}
\ket{V_{[6,6]}^{T}(s_2,\bm{s}_4,s_6,s_2,s_4)} .
\]
In this case, $V_{[5]}$ is a $5\times 5$ Vandermonde matrix and is therefore
invertible. Applying $U_{V_{[5]}^{-1}}$ to the first two registers gives
\[
\ket{\bm{s}_1\bm{s}_2\bm{s}_3\bm{s}_4\bm{s}_5\bm{s}_6}
\sum_{\bm{r}\in \mathbb{F}_7^6}
\ket{\bm{r}_1,\bm{r}_3)}\ket{\bm{r}_2,\bm{r}_4)}
\ket{V_{[6,6]}^{T}(s_1,s_3,s_5,r_1,r_3)}
\ket{V_{[6,6]}^{T}(s_2,s_4,s_6,r_2,r_4)} .
\]

To further disentangle the remaining registers, let
$W_3=
\begin{bmatrix}
1 & 0 & 0 & 0 & 0\\
0 & 1 & 0 & 0 & 0\\
0 & 0 & 1 & 0 & 0\\
0 & 0 & 0 & 1 & 0\\
\multicolumn{5}{c}{V_{[6,6]}^{T}}
\end{bmatrix}.$
The combiner applies two unitary operations $U_{W_3}$ to the registers
$\ket{s_1,s_3,s_5,r_1,r_3}$ and $\ket{s_2,s_4,s_6,r_2,r_4}$, respectively. The
state is then transformed into
\[
\begin{aligned}
&\ket{\bm{s}_1\bm{s}_2\bm{s}_3\bm{s}_4\bm{s}_5\bm{s}_6}
\sum_{(r_1,r_2)\in\mathbb{F}_7^2}\ket{\bm{r}_1}\ket{\bm{r}_2}
\otimes
\sum_{r_3\in\mathbb{F}_7}
\ket{\bm{V}_{[6,6]}^{T}(\bm{s}_1,\bm{s}_3,\bm{s}_5,\bm{r}_1,\bm{r}_3)}
\ket{V_{[6,6]}^{T}(s_1,s_3,s_5,r_1,r_3)} \\
&\qquad\qquad\qquad\qquad\qquad\qquad\quad \otimes
\sum_{r_4\in\mathbb{F}_7}
\ket{\bm{V}_{[6,6]}^{T}(\bm{s}_2,\bm{s}_4,\bm{s}_6,\bm{r}_2,\bm{r}_4)}
\ket{V_{[6,6]}^{T}(s_2,s_4,s_6,r_2,r_4)} .
\end{aligned}
\]
\noindent\textbf{Case 3: \(k=4\).}
Let the participating set be
$\mathcal{A}=\{P_1,P_2,P_3,P_4\}.$
In this case, the quantum state can be represented in two layers as
\begin{align*}
|\Psi_4\rangle
=
\sum_{r\in\mathbb{F}_7^4}
&\left|\bm{V}_{[4]}^{T}(\bm{s}_1,\bm{s}_3,\bm{s}_5,\bm{r}_1,\bm{r}_3)\right\rangle
 \left|\bm{V}_{[4]}^{T}(\bm{s}_2,\bm{s}_4,\bm{s}_6,\bm{r}_2,\bm{r}_4)\right\rangle       \\
&\otimes
 \left|V_{[5,6]}^{T}(s_1,s_3,s_5,r_1,r_3)\right\rangle
 \left|V_{[5,6]}^{T}(s_2,s_4,s_6,r_2,r_4)\right\rangle ,
\end{align*}
and
\[
\sum_{r\in\mathbb{F}_7^4}
\left|\bm{V}_{[4]}^{T}(0,\bm{r}_3,\bm{r}_4,\bm{r}_5,\bm{r}_6)\right\rangle
\left|V_{[5,6]}^{T}(0,r_3,r_4,r_5,r_6)\right\rangle .
\]

The combiner first applies the unitary operation corresponding to
\((V_{[4]}^{T})^{-1}\) to
\[
\sum_{r\in\mathbb{F}_7^4}
\left|\bm{V}_{[4]}^{T}(0,\bm{r}_3,\bm{r}_4,\bm{r}_5,\bm{r}_6)\right\rangle
\left|V_{[5,6]}^{T}(0,r_3,r_4,r_5,r_6)\right\rangle .
\]
This yields
\[
\sum_{r\in\mathbb{F}_7^4}
|r_3,r_4,r_5,r_6\rangle
\left|V_{[5,6]}^{T}(0,r_3,r_4,r_5,r_6)\right\rangle .
\]

Next, the combiner applies \(V_{[4]}^{T}\) to
$(0,0,0,0,r_3)
\quad \text{and} \quad
(0,0,0,0,r_4),$
and then uses \(\operatorname{ADD}^{\dagger}\) gates to eliminate the coherence
between the random variables \(r_3,r_4\) and the secret variables. This gives
\begin{equation}\label{eq.38}
\begin{aligned}
|\Psi_4'\rangle
=
\sum_{r\in\mathbb{F}_7^4}
&\left|(\bm{V}_{[4]}^{[4]})^{T}(\bm{s}_1,\bm{s}_3,\bm{s}_5,\bm{r}_1)\right\rangle
 \left|(\bm{V}_{[4]}^{[4]})^{T}(\bm{s}_2,\bm{s}_4,\bm{s}_6,\bm{r}_2)\right\rangle    \\
&\otimes
 \left|V_{[5,6]}^{T}(s_1,s_3,s_5,r_1,r_3)\right\rangle
 \left|V_{[5,6]}^{T}(s_2,s_4,s_6,r_2,r_4)\right\rangle .
\end{aligned}
\tag{38}
\end{equation}

Since \((V_{[4]}^{[4]})^{T}\) is a \(4\times 4\) Vandermonde matrix, it is
invertible. Applying the unitary operation corresponding to
\(\bigl((V_{[4]}^{[4]})^{T}\bigr)^{-1}\) to Eq.~\eqref{eq.38}
gives
\begin{equation}\label{eq:39}
\begin{aligned}
|\Psi_4''\rangle
=
|\bm{s}_1\bm{s}_2\bm{s}_3\bm{s}_4\bm{s}_5\bm{s}_6\rangle
\sum_{r\in\mathbb{F}_7^4}
& |\bm{r}_1\rangle |\bm{r}_2\rangle                                      
\left|V_{[5,6]}^{T}(s_1,s_3,s_5,r_1,r_3)\right\rangle\\
 &\otimes
 \left|V_{[5,6]}^{T}(s_2,s_4,s_6,r_2,r_4)\right\rangle .
\end{aligned}
\tag{39}
\end{equation}

Define
$W_4=
\begin{bmatrix}
I_{3\times 3} & 0\\
0 & V_{[5,6]}^{T}
\end{bmatrix}.$
Using the second-layer particles \(|r_3\rangle\) and \(|r_4\rangle\), the
combiner applies the unitary operation \(U_{W_4}\) to the registers
$|s_1,s_3,s_5,r_1,r_3\rangle
~ \text{and} ~
|s_2,s_4,s_6,r_2,r_4\rangle $.

As a result, one obtains a uniform superposition
$\sum_{f_1\in\mathbb{F}_7}|f_1\rangle$
independent of \(s_1,s_3,s_5\), and similarly a uniform superposition
$\sum_{f_2\in\mathbb{F}_7}|f_2\rangle$
independent of \(s_2,s_4,s_6\). After a change of variables and rewriting, we obtain
\[
|\Psi_4\rangle
=
|\bm{s}_1\bm{s}_2\bm{s}_3\bm{s}_4\bm{s}_5\bm{s}_6\rangle
\sum_{f_1\in\mathbb{F}_7}|\bm{f}_1\rangle |f_1\rangle
\sum_{f_2\in\mathbb{F}_7}|\bm{f}_2\rangle |f_2\rangle .
\]

At this point, the classically pre-encrypted secret state has been completely
disentangled from the remaining part of the system. By linearity, the same
reconstruction procedure also applies to an arbitrary superposition state.

In the classical reconstruction phase, assume without loss of generality that
the authorized set is
$\{P_1,P_2,P_4,P_6\}.$
The corresponding Lagrange coefficients over $\mathbb{F}_7$ are
$(\lambda_1,\lambda_2,\lambda_4,\lambda_6)=(6,4,1,4).$
By using Lagrange interpolation, the authorized set can recover the classical
key
$K=(\bm{a},\bm{b}).$
Then Peter applies the inverse pre-encryption operation
$U_K^{-1}
=
\bigotimes_{\ell=1}^{6} Z_\ell^{-b_\ell}X_\ell^{-a_\ell}$
to the quantum state $\ket{\bm{s}}$, and obtains the original secret quantum
state $\ket{\bm{m}}.$

\begin{remark}
This example shows that intermediate-set leakage may indeed occur in the original
Non-PQSS scheme constructed from ECSS codes. To address this problem, the present
paper combines the original Non-PQSS scheme with a classical $(4,6)$ threshold
sharing mechanism. The classical threshold-sharing layer masks the information
that may be leaked to intermediate sets, thereby ensuring that unauthorized sets
obtain no information about the original secret. Only participant sets satisfying
the authorization condition can recover both the required quantum shares and the
classical pre-encryption key, and can finally reconstruct the original secret
state by applying the inverse transformation. Consequently, the original
Non-PQSS scheme is transformed into a perfect quantum secret sharing scheme,
while retaining the communication efficiency of the original scheme in the
quantum reconstruction phase.
\end{remark}

\section{Performance Analysis}
\subsection{Cost Model and Evaluation Metrics}

To evaluate the communication efficiency of the proposed scheme, we mainly use
the amount of quantum communication as the primary performance metric, measured
uniformly in terms of the number of qudits. The reason for focusing on quantum
communication is that, in quantum secret sharing, quantum channel resources are
usually much more expensive than classical channel resources. Since the proposed
scheme introduces a classical threshold-sharing layer on top of the original
communication-efficient QSS layer to protect the pre-encryption key, this
additional layer inevitably brings extra classical communication overhead.
However, in order to remain consistent with the quantum communication cost
metrics commonly used in the existing CE-QSS literature, the cost analysis below
only counts the quantum communication cost associated with the considered
authorized set.

Let the secret size be $m$ qudits, and let $n$ denote the total number of
participants. In the reconstruction phase, suppose that the authorized set
participating in secret recovery has size $k$, where
$t\leq k\leq n.$
Let $B_{\mathrm{test}}$ denote the number of test code blocks used in the
distribution phase, and let $N$ denote the number of decoy particles inserted by
each participant in the reconstruction phase. In addition, $v_1$ and $v_2$
represent the numbers of first-layer and second-layer information particles
received by each participant in the distribution phase, respectively.

The quantum communication efficiency of the proposed scheme is denoted by
$\eta$ and is defined as
$\eta=\frac{c}{q},$
where $c$ denotes the total number of qudits contained in the secret, and $q$
denotes the total number of qudits transmitted in the scheme. In this paper, the
communication cost only counts the quantum communication associated with the
considered authorized set, measured in qudits. The classical communication
introduced by the classical threshold-sharing layer is not included in this
quantum communication cost metric.
\subsection{Quantum Communication Cost of $k$ Participants in Different Phases}

Under the pre-specified reconstruction-set model described in Section~4.1, let
the authorized set used for secret recovery in the current round have size $k$,
where
$t\leq k\leq n.$
In the distribution phase, Alice prepares the encoded secret state and
$B_{\mathrm{test}}$ test code blocks, and distributes the information particles
to the $k$ selected participants. Specifically, Alice sends
$k(v_1+v_2)$ information particles and $kB_{\mathrm{test}}a_1$ test particles.
Thus, the number of particles required in the distribution phase is
$k(v_1+v_2)+kB_{\mathrm{test}}a_1.$

In the reconstruction phase, the participants
$\mathrm{Bob}_{i_1},\mathrm{Bob}_{i_2},\ldots,\mathrm{Bob}_{i_k}$ prepare a total
of $kN$ decoy particles. Therefore, the total number of particles required by the
proposed scheme is
$k(v_1+v_2)+kB_{\mathrm{test}}a_1+kN.$
Accordingly, the quantum communication efficiency is given by
\begin{equation}
\eta
=
\frac{m}{k(v_1+v_2)+kB_{\mathrm{test}}a_1+kN}.
\tag{40}
\label{eq.40}
\end{equation}

Suppose first that the secret is reconstructed by a set of $d_i$ participants, where $n\geq d_i\geq d$. Under the pre-announced reconstruction model, Alice distributes quantum systems only to these $d_i$ participants. Hence, the distribution phase involves $d_i(v_1+v_2)$ information particles and $d_iB_{\mathrm{test}}a_1$ test particles, and the total number of transmitted particles in this phase is
$d_i(v_1+v_2)+d_iB_{\mathrm{test}}a_1 .$
During the reconstruction phase, the same $d_i$ participants send $d_iv_1$ information particles together with $d_iN$ decoy particles to Peter. Therefore, the reconstruction-phase quantum transmission is $d_iv_1+d_iN .$

Now suppose that the reconstruction is performed by $t_i$ participants with $d>t_i\geq t$, and that the secret consists of $m$ qudits. In the distribution phase, Alice transmits $t_i(v_1+v_2)$ information particles and $t_iB_{\mathrm{test}}a_1$ test particles to the selected participants. Thus, the total number of particles distributed in this phase is
$t_i(v_1+v_2)+t_iB_{\mathrm{test}}a_1 $.
In the reconstruction phase, the $t_i$ participants transmit $t_i(v_1+v_2)$ information particles and $t_iN$ decoy particles to Peter, resulting in a total of
$t_i(v_1+v_2)+t_iN$
transmitted particles.

In the following analysis, the communication cost refers to the number of quantum particles transmitted in the corresponding phase.
Based on the above cost model, the average quantum communication cost required
to recover each secret qudit can be evaluated from the following perspectives:
\begin{enumerate}
    \item $C_1$: the average quantum communication cost per recovered secret
    qudit in both the distribution and reconstruction phases;

    \item $C_2$: the average quantum communication cost per recovered secret
    qudit in the secret reconstruction phase;

    \item $C_3$: the average quantum communication cost per recovered secret
    qudit in the secret reconstruction phase, excluding decoy particles.
\end{enumerate}

\subsection{Quantum Communication Cost}

Under the pre-designated authorized-set model, suppose the authorized set size is $k$.

In the distribution phase, Alice transmits $k(v_1+v_2)$ information particles and
$k B_{\text{test}} a_1$ test particles.

In the reconstruction phase, participants send $k(v_1+v_2)$ information particles
and $kN$ decoy particles.
Thus, the total number of transmitted particles is
$k(v_1+v_2) + k B_{\text{test}} a_1 + kN,$
and the communication efficiency is given by
$\eta = \frac{m}{k(v_1+v_2) + k B_{\text{test}} a_1 + kN}.$

For $d_i$ participants, where $n \ge d_i \ge d$, the total number of particles is
$d_i(v_1+v_2) + d_i B_{\text{test}} a_1 + d_i N.$

For $t_i$ participants, where $d > t_i \ge t$, the total number is
$t_i(v_1+v_2) + t_i B_{\text{test}} a_1 + t_i N.$

Accordingly, the quantum communication cost can be evaluated under five metrics.
The resulting expressions are summarized in Table~\ref{tab:communication_cost}.
\begin{table}[htbp]
\centering
\caption{Communication efficiency and communication costs of the proposed
$\left((t,n,d,z)\right)$ scheme}
\label{tab:communication-cost}

\scriptsize
\renewcommand{\arraystretch}{2.2}
\setlength{\tabcolsep}{4pt}

\begin{tabular}{c|c|c}
\hline
\textbf{Metric}
&
\textbf{$d_i$ participants}
&
\textbf{$t_i$ participants}
\\
\hline
Secret size
&
$m$
&
$m$
\\
\hline
$\eta$
&
$\displaystyle
\frac{m}{d_i(v_1+v_2)+d_iB_{\mathrm{test}}a_1+d_iN}
$
&
$\displaystyle
\frac{m}{t_i(v_1+v_2)+t_iB_{\mathrm{test}}a_1+t_iN}
$
\\
\hline
$C_1$
&
$\displaystyle
\frac{d_i(v_1+v_2)+d_iB_{\mathrm{test}}a_1+d_iv_1+d_iN}{m}
$
&
$\displaystyle
\frac{2t_i(v_1+v_2)+t_iB_{\mathrm{test}}a_1+t_iN}{m}
$
\\
\hline
$C_2$
&
$\displaystyle
\frac{d_iv_1+d_iN}{m}
$
&
$\displaystyle
\frac{t_i(v_1+v_2)+t_iN}{m}
$
\\
\hline
$C_3$
&
$\displaystyle
\frac{d_iv_1}{m}
$
&
$\displaystyle
\frac{t_i(v_1+v_2)}{m}
$
\\
\hline
\end{tabular}

\vspace{1mm}
\footnotesize
\noindent
Note: $d\leq d_i<n$, $t\leq t_i<d$, and $d_i,t_i$ are integers.
\end{table}

We compare the communication cost $C_3$ of the same scheme under different 
numbers of participating parties. Specifically, when the combiner contacts 
$d$ participants for secret reconstruction, the communication cost per secret 
qudit is
\[
CC_n(d)
=
\frac{d v_1}{m}
=
\frac{\dfrac{d(t-z)}{\gcd(t-z,d-t)}}{m}.
\]
On the other hand, when only $t$ participants are contacted, the corresponding
communication cost is
\[
CC_n(t)
=
\frac{t(v_1+v_2)}{m}
=
\frac{\dfrac{t(d-z)}{\gcd(t-z,d-t)}}{m}.
\]
Since
$d(t-z)<t(d-z),$
we have
\[
CC_n(d)
=
\frac{\dfrac{d(t-z)}{\gcd(t-z,d-t)}}{m}
<
\frac{\dfrac{t(d-z)}{\gcd(t-z,d-t)}}{m}
=
CC_n(t).
\]

Therefore, within the same scheme, when the combiner reconstructs the secret by 
contacting $d$ participants, the quantum communication cost required per unit 
secret is lower than that required when contacting only $t$ participants. That 
is,
$CC_n(d)<CC_n(t),$
which demonstrates the communication efficiency of the proposed scheme in the 
secret reconstruction phase.
\subsection{Communication Efficiency Analysis When $d$ Participants Recover Secrets of Different Sizes}

We now analyze the communication efficiency when different numbers of 
participants are involved in secret reconstruction. For fixed $t$ and $z$, since 
$n=t+z$, the parameters of the $\left((t,n,d;z)\right)$ quantum secret sharing 
scheme are determined once $d$ is specified. Moreover, the secret size is
$m=a_1v_1=\operatorname{lcm}(d-z,t-z).$
When the secret is recovered by $d$ participants, different values of $d$ may 
correspond to different recoverable secret sizes. We therefore analyze the effect 
of $d$ on $C_1$, $C_2$, $C_3$, and $\eta$ in the proposed scheme.

Assume that a secret of size $m_i$ can be efficiently recovered by $d$ 
participants. Since
$m_i=a_1v_1,$
the reconstruction communication cost excluding decoy particles is
$C_3
=
\frac{dv_1}{a_1v_1}
=
\frac{d}{a_1}.$
By the parameter relation $a_1=d-z,$ we have
$C_3
=
\frac{d}{d-z}
=
\frac{1}{1-\frac{z}{d}}.$
Thus, for fixed $z$, $C_3$ decreases as $d$ increases. This shows that 
contacting more participants can reduce the reconstruction communication cost 
per secret qudit when decoy particles are not counted.

Next, suppose that $\gcd(t-z,d-t)$ remains fixed while $d$ increases. In this 
case,
$v_1=\frac{t-z}{\gcd(t-z,d-t)}$
is also fixed. The reconstruction communication cost including decoy particles is
\[
C_2
=
\frac{dv_1+dN}{a_1v_1}
=
\left(1+\frac{N}{v_1}\right)C_3.
\]
Since $v_1$ is constant under the above assumption, the factor 
$\left(1+\frac{N}{v_1}\right)$ is also constant. Therefore, $C_2$ decreases 
with the increase of $d$ as well.

We now analyze the relationship between $C_1$ and $d$. We have
\begin{align}
C_1
&=
\frac{d(v_1+v_2)+dB_{\mathrm{test}}a_1+dv_1+dN}{m}
\nonumber\\
&=
\frac{
d\dfrac{d-z}{\gcd(t-z,d-t)}
+dB_{\mathrm{test}}(d-z)
+d\dfrac{t-z}{\gcd(t-z,d-t)}
+dN
}{
(d-z)\dfrac{t-z}{\gcd(t-z,d-t)}
}
\nonumber\\
&=
\frac{
\bigl(1+B_{\mathrm{test}}\gcd(t-z,d-t)\bigr)(d-z)
+\bigl(1+B_{\mathrm{test}}\gcd(t-z,d-t)\bigr)z
}{
t-z
}
\nonumber\\
&\quad
+
\frac{
\bigl(N\gcd(t-z,d-t)+(t-z)\bigr)
\left(1+\dfrac{z}{d-z}\right)
}{
t-z
}.
\tag{41}
\label{eq.41}
\end{align}

By applying the basic inequality to Eq.~(\ref{eq.41}), we obtain
\[
C_1
\ge
2
\sqrt{
\bigl(1+B_{\mathrm{test}}\gcd(t-z,d-t)\bigr)
\bigl(N\gcd(t-z,d-t)+(t-z)\bigr)z
}.
\]

Equality holds if and only if
\[
d-z
=
\sqrt{
\frac{
\bigl(N\gcd(t-z,d-t)+(t-z)\bigr)z
}{
1+B_{\mathrm{test}}\gcd(t-z,d-t)
}
}
\ge
\sqrt{
\frac{
N\gcd(t-z,d-t)z
}{
\bigl(1+B_{\mathrm{test}}\bigr)\gcd(t-z,d-t)
}
}
=
\sqrt{\frac{Nz}{1+B_{\mathrm{test}}}}.
\]
We now analyze the relationship between $\eta$ and $d$. We have
\[
\begin{aligned}
\eta
&=
\frac{m}{d(v_1+v_2)+dB_{\mathrm{test}}a_1+dN}
=
\frac{
(d-z)\dfrac{t-z}{\operatorname{gcd}(t-z,d-t)}
}{
d\dfrac{d-z}{\operatorname{gcd}(t-z,d-t)}
+dB_{\mathrm{test}}(d-z)+dN
} \\[2mm]
&=
\frac{
\dfrac{t-z}{\operatorname{gcd}(t-z,d-t)}
}{
(d-z)\left(\dfrac{1}{\operatorname{gcd}(t-z,d-t)}+B_{\mathrm{test}}\right)
+z\left(\dfrac{1}{\operatorname{gcd}(t-z,d-t)}+B_{\mathrm{test}}\right)
+\dfrac{z}{d-z}N+N
}.
\end{aligned}
\]
The corresponding equality condition is
\[
d-z
=
\sqrt{
\frac{N\,\operatorname{gcd}(t-z,d-t)\,z}
{1+B_{\mathrm{test}}\operatorname{gcd}(t-z,d-t)}
}.
\]
Moreover,
\[
\sqrt{
\frac{N\,\operatorname{gcd}(t-z,d-t)\,z}
{1+B_{\mathrm{test}}\operatorname{gcd}(t-z,d-t)}
}
\geq
\sqrt{
\frac{N\,\operatorname{gcd}(t-z,d-t)\,z}
{(1+B_{\mathrm{test}})\operatorname{gcd}(t-z,d-t)}
}
=
\sqrt{
\frac{Nz}{1+B_{\mathrm{test}}}
}.
\]
In the secret distribution phase, the sequence sent from Alice to $\mathrm{Bob}_i$
contains information particles and $B_{\mathrm{test}}$ test blocks. In the secret
reconstruction phase, $\mathrm{Bob}_i$ inserts $N$ decoy particles into the
particle sequence transmitted to Peter. To prevent eavesdropping, $N$ should be
chosen to be sufficiently large. In this case, we have
\[
N\gg B_{\mathrm{test}}.
\]
Hence,
\[
d-z<\sqrt{\frac{Nz}{1+B_{\mathrm{test}}}}.
\]
Therefore, $C_1$ decreases as $d$ increases, while the communication efficiency
$\eta$ increases as $d$ increases.

In particular, when $z=t-1$, the scheme reduces to a $(t,n)$ threshold scheme.
In this case, as $d$ increases, the costs $C_1$, $C_2$, and $C_3$ decrease,
whereas the communication efficiency $\eta$ increases.

This section analyzes the proposed scheme mainly from two aspects: communication
cost and communication efficiency. First, the amount of quantum communication is
used as the primary evaluation metric, and the quantum resources required in the
secret distribution phase, the detection phase, and the secret reconstruction
phase are discussed separately. This provides a more comprehensive
characterization of the overall communication overhead of the scheme. Second,
under different numbers of authorized participants involved in reconstruction,
we analyze the variation of the reconstruction communication cost. The results
show that, as the number of participants involved in reconstruction increases,
the combiner can obtain fewer quantum shares from each participant, thereby
reducing the quantum communication cost required per unit secret. Finally,
although the proposed scheme introduces a classical threshold sharing mechanism
and a detection mechanism, the additional overhead mainly consists of classical
communication and detection-related costs, and does not destroy the
communication efficiency of the original scheme in the quantum reconstruction
phase. Therefore, the analysis in this section shows that the proposed scheme
enhances security while still maintaining good quantum communication efficiency,
which indicates its potential practical value.

\section{Conclusion}

This paper has proposed a class of communication-efficient perfect quantum
secret sharing schemes to address the potential risk of intermediate-set
information leakage in existing ECSS-based communication-efficient Non-PQSS
schemes. Compared with the original schemes, the advantages of the proposed
construction are mainly reflected in the following three aspects.

First, this paper provides an intermediate-set detection method. By representing
the quantum shares held by a participant set as linear combinations of secret
variables and random variables, and by performing elimination analysis through
invertible linear transformations, one can determine whether certain participant
sets are able to obtain partial information about the secret. This method
reveals the concrete mechanism of intermediate-set leakage and provides a
theoretical basis for subsequent security enhancement.

Second, the proposed scheme preserves the low communication-cost advantage of
the original communication-efficient quantum secret sharing construction. The
scheme allows the combiner to choose different reconstruction strategies
according to the number of participants involved in secret recovery. When more
authorized participants participate in reconstruction, each participant only
needs to send a portion of his or her quantum share, thereby reducing the
quantum communication cost in the secret reconstruction phase.

Finally, this paper introduces a classical threshold sharing mechanism to
protect the quantum pre-encryption key. As a result, even if an intermediate set
obtains partial information about the encrypted quantum state, it cannot recover
the original quantum secret. Only authorized sets can simultaneously recover the
required quantum shares and the classical key, and hence complete the secret
reconstruction. Therefore, the proposed scheme transforms the original
non-perfect quantum secret sharing scheme into a perfect quantum secret sharing
scheme.

In summary, the proposed scheme simultaneously achieves enhanced security
against intermediate-set leakage, improved communication efficiency, and a
perfectness-preserving construction. It therefore provides a feasible approach
for designing secure and communication-efficient quantum secret sharing
protocols.

\section*{Acknowledgment}
{
This work is supported by the National Natural Science Foundation of China (NSFC) under Grant No.12471437.
}



\backmatter






\bibliography{sn-bibliography}

\end{document}